\documentclass[a4paper,UKenglish,cleveref, autoref, thm-restate]{lipics-v2021}

\hideLIPIcs  

\title{Double Exponential Lower Bound for \textsc{Telephone Broadcast}}

\author{Prafullkumar Tale}{Indian Institute of Science Education and Research Bhopal, Bhopal, India \and \url{https://pptale.github.io/}}{prafullkumar@iiserb.ac.in}{https://orcid.org/0000-0001-9753-0523}{}

\authorrunning{Prafullkumar Tale} 

\Copyright{Prafullkumar Tale} 



\nolinenumbers 

\EventEditors{John Q. Open and Joan R. Access}
\EventNoEds{2}
\EventLongTitle{42nd Conference on Very Important Topics (CVIT 2016)}
\EventShortTitle{CVIT 2016}
\EventAcronym{CVIT}
\EventYear{2016}
\EventDate{December 24--27, 2016}
\EventLocation{Little Whinging, United Kingdom}
\EventLogo{}
\SeriesVolume{42}
\ArticleNo{23}

\usepackage{amsmath}
\usepackage{amssymb}
\usepackage{graphicx}
\usepackage{complexity} 
\usepackage{algorithm}
\usepackage[noend]{algpseudocode}
\usepackage{footnote}
\algnewcommand{\LeftComment}[1]{\Statex \(\triangleright\) #1}

\usepackage[svgnames]{xcolor}

\newcommand{\vc}{\mathtt{vc}}

\newcommand{\td}{\mathtt{td}}

\newcommand{\dist}{\mathsf{dist}}
\newcommand{\size}{\mathsf{size}}

\newcommand{\calA}{\mathcal{A}}

\newcommand{\calO}{\ensuremath{{\mathcal O}}}

\newcommand{\calT}{\mathcal{T}}
\newcommand{\calV}{\mathcal{V}}

\newcommand{\true}{\texttt{True}}
\newcommand{\false}{\texttt{False}}

\newcommand{\ETH}{\textsf{ETH}}
\newcommand{\hard}{\textsf{hard}}
\newcommand{\complete}{\textsf{complete}}
\newcommand{\para}{\textsf{para}}

\newcommand{\yes}{\textsc{Yes}}
\newcommand{\no}{\textsc{No}}

\newtheorem{reduction-rule}{Reduction Rule}
\newtheorem{branching-rule}{Branching Rule}
\newtheorem{branching-rule1}{Branching Rule}
\newtheorem{reduction rule}{Reduction Rule}
\newtheorem{preprocessing rule}{Preprocessing Rule}
\newtheorem{branching rule}{Branching Rule}
\newtheorem{marking-scheme}{Marking Scheme}

\algrenewcommand\algorithmicrequire{\textbf{Input:}}
\algrenewcommand\algorithmicensure{\textbf{Task:}}

\newcommand{\defproblem}[3]{
  \vspace{1mm}
\noindent\fbox{
  \begin{minipage}{0.96\textwidth}
  \begin{tabular*}{\textwidth}{@{\extracolsep{\fill}}lr} #1 \\ \end{tabular*}
  {\bf{Input:}} #2  \\
  {\bf{Question:}} #3
  \end{minipage}
  }
  \vspace{1mm}
}

\usepackage{hyperref}
\usepackage{thmtools}
\usepackage{thm-restate}

\begin{document}

\maketitle

\begin{abstract}
Consider the \textsc{Telephone Broadcast} 
problem in which an input is a connected graph $G$ on $n$ vertices, 
a source vertex $s \in V(G)$, and a positive integer $t$.
The objective is to decide whether there is a broadcast protocol 
from $s$ that ensures that all the vertices of $G$ 
get the message in at most $t$ rounds.
We consider the broadcast protocol where, in a round, any node aware of the message can forward it to  at most one of its neighbors.
As the number of nodes aware of the message can at most double at each 
round, for a non-trivial instance we have $n \le 2^t$.
Hence, the brute force algorithm that checks all the permutations of
the vertices runs in time $2^{2^{\calO(t)}} \cdot n^{\calO(1)}$.
As our first result, we prove this simple algorithm is the best possible
in the following sense.
\begin{itemize}
\item \textsc{Telephone Broadcast} does not admit an algorithm running in
time $2^{2^{o(t)}} \cdot n^{\calO(1)}$, unless the \ETH\ fails.
\end{itemize}
To the best of our knowledge, this is only the fourth example
of \NP-Complete problem that admits a double exponential
lower bound when parameterized by the solution size.
It also resolves the question by 
Fomin, Fraigniaud, and Golovach [WG 2023].
In the same article, the authors asked whether the 
problem is \FPT\ when parameterized by 
the feedback vertex set number of the graph.
We answer this question in the negative.
\begin{itemize}
\item \textsc{Telephone Broadcast}, when restricted to graphs
of the feedback vertex number one, and hence
treewidth of two, is \NP-\complete.
\end{itemize}
We find this a relatively rare example of problems that admit
a polynomial-time algorithm on trees but is 
\NP-\complete\ on graphs of treewidth two.
\end{abstract}

\section{Introduction}
\label{sec:intro}

The aim of broadcasting in a network is to transmit a message 
from a given source vertex of the network to all the other vertices.
With the increasing interest in interconnection networks, 
there is considerable research dedicated to broadcasting
considering different models.
These models can have different attributes such as 
the number of sources, 
the number of vertices a particular vertex can forward message to 
in each round, 
the distances of each call,  
the number of destinations, etc., to name a few. 
In this paper, we consider the classical model of broadcasting~\cite{DBLP:journals/networks/HedetniemiHL88}
in which the process is split into discrete time steps and
at each time step, i.e., round, an informed vertex can forward the message 
to at most one of its uninformed neighbors.
There is a unique vertex, called \emph{source}, that has the 
message at the start of the process.
For a more detailed introduction to broadcasting, we refer the reader to
surveys like \cite{fraigniaud1994methods,hromkovivc1996dissemination}
and on recent PhD thesis~\cite{khanlari2023broadcasting}.

Formally, consider a connected simple graph $G$  
and let $s \in V(G)$ be the unique source of a message $M$.
At any given round, any node $u \in V(G)$ aware of $M$ can forward 
$M$ to at most one neighbor $v$ of $u$.
The minimum number of rounds for broadcasting a message 
from $s$ in $G$ to all other vertices is denoted by $b(G, s)$.
The problem of computing the broadcast time $b(G, s)$ for 
a given graph $G$ and a given source $s \in V(G)$ 
is \NP-\hard~\cite{DBLP:journals/siamcomp/SlaterCH81}. 
Also, the results of \cite{DBLP:journals/jacm/PapadimitriouY82}
imply that it is \NP-\complete\ to 
decide whether $b(G, s) \le t$ for graphs with 
$n = 2t$ vertices.
The problem is known to be \NP-\complete\ for restricted graph classes
such as $3$-regular planar graphs~\cite{DBLP:journals/ipl/Middendorf93}
and split graphs~\cite{DBLP:journals/tcs/JansenM95}.
Polynomial-time algorithms are known for trees \cite{DBLP:journals/siamcomp/SlaterCH81} and 
some classes of tree-like graphs 
\cite{DBLP:journals/jco/CevnikZ17,DBLP:journals/join/GholamiHM23,DBLP:journals/jco/HarutyunyanM08}. 

These hardness results motivate to study approximation, exact exponential
and parameterized algorithms for the problem.
For example, the polynomial time algorithm in \cite{DBLP:journals/siamdm/KortsarzP95} computes, for every graph $G$ and every source $s$, 
a broadcast protocol from $s$ performing  in 
$2\cdot b(G, s) + \calO(\sqrt{n})$ 
round, and hence, this algorithm has an approximation ratio of $2+o(1)$ 
for graphs with broadcast time $\gg \sqrt{n}$,
but $\tilde{\Theta}(\sqrt{n})$ in general.
Later, a series of papers tighten the approximation ratio from 
$\calO(\log^2 n / \log \log n)$ \cite{DBLP:conf/focs/Ravi94}, 
to $\calO(\log n)$ \cite{DBLP:journals/siamcomp/Bar-NoyGNS00}, 
and 
eventually $\calO(\log n/ \log \log n)$ \cite{DBLP:journals/jcss/ElkinK06}. 
For special graph classes, algorithms with better approximation results
are known~\cite{DBLP:conf/caldam/BhabakH15,DBLP:journals/join/BhabakH19,DBLP:conf/ciac/HarutyunyanH23}.

Fomin et al.~\cite{DBLP:conf/wg/FominFG23} presented an exact exponential
algorithm that runs in time $3^n \cdot n^{\calO(1)}$.
This improves upon a brute force algorithm that checks all the permutations
of the vertices where in each round, an informed vertex 
forwards the message to the first uninformed neighbor on its right. 
They also initiated the parameterized complexity study of the problem.
We mention the relevant definitions in Section~\ref{sec:prelims} but state
the formal definition of the problem before proceeding.

\defproblem{\textsc{Telephone Broadcast}}{A connected graph $G$ on $n$ vertices, a source vertex $s \in V(G)$, and a positive integer $t$.}{Does there exist a broadcasting protocol 
(where an informed vertex can forward the message to at most one of 
its uninformed neighbors) that starts from $s$ and informs all the vertices in time $t$?}

Fomin et al.~\cite{DBLP:conf/wg/FominFG23} observed that 
the problem admits a trivial kernel with $2^{t}$ vertices
when parameterized by the natural parameter, i.e.,  
the number of rounds $t$. 
To see this, note that the number of informed nodes 
can at most double at each round.
This implies that at most $2^t$ vertices can have received the message after 
$t$ communication rounds.
Hence, if $n > 2^t$, it is safe to conclude that it is a \no-instance.
Combined with the exact exponential algorithms mentioned above,
this implies that the \textsc{Telephone Broadcast} admits
an algorithm running in time $2^{2^{\calO(t)}} \cdot n^{\calO(1)}$.
Fomin et al.~\cite{DBLP:conf/wg/FominFG23} explicitly asked
whether it is possible to obtain an algorithm with improved running time.
We answer this question in the negative.

\begin{theorem}
\label{thm:sol-size-lb}
Unless the \ETH\ fails,
\textsc{Telephone Broadcast} does not admit
an algorithm running in time $2^{2^{o(t)}} \cdot n^{\calO(1)}$.
\end{theorem}

To the best of our knowledge, this is only the fourth such 
example of tight double-exponential lower bound when 
parameterized by the solution size. 
The other three examples are \textsc{Edge Clique Cover}~\cite{DBLP:journals/siamcomp/CyganPP16}, \textsc{BiClique Cover}~\cite{DBLP:conf/iwpec/ChandranIK16},
and
\textsc{Test Cover}~\cite{DBLP:journals/corr/abs-2402-08346}
\footnote{We mention a result in~\cite{DBLP:journals/algorithmica/AboulkerBKS23} which states that \textsc{Grundy Coloring} does not admit an 
algorithm running in time $f(k) \cdot n^{o(2^{k - \log k})}$, unless 
the \ETH\ fails.}.
As in these examples, a simple corollary of the above theorem is
that \textsc{Telephone Broadcast}
does not admit a kernelization algorithm that returns a kernel with 
$2^{o(t)}$ many vertices.

Towards the structural parameterization of the problem,
Fomin et al.~\cite{DBLP:conf/wg/FominFG23}
proved that the problem admits fixed-parameter tractable
algorithm when parameterized by 
the vertex cover number and the cyclomatic number of the input graph.
They mentioned that it is interesting to consider other structural parameterizations and, in particular, asked whether the problem
admits fixed-parameter tractable algorithms 
when parameterized by the feedback vertex number or treewidth, which are 
`smaller' parameters.
Our following result states that the problem is \para-\NP-\hard\ when 
parameterized by these parameters and hence highly unlikely
to admit such algorithms.

\begin{theorem}
\label{thm:fvs-np-hard}
\textsc{Telephone Broadcast}, when restricted to graphs
\begin{itemize}
\item with the feedback vertex set number one, and hence treewidth of two, and
\item with the pathwidth of three,
\end{itemize}
remains \NP-\complete.
\end{theorem}
This results prove polynomial vs \NP-\complete\ dichotomy separating treewidth one from treewidth two, 
which is a rare phenomenon and to the best of our knowledge, the only 
other examples are \textsc{Cutwidth}~\cite{DBLP:journals/jacm/Yannakakis85,DBLP:journals/tcs/MonienS88}, 
\textsc{L(2,1)-labeling}~\cite{DBLP:conf/icalp/FialaGK05}, 
\textsc{List L(1, 1)-Labeling}, \textsc{L(1, 1)-Prelabeling Extension}~\cite{DBLP:journals/tcs/FialaGK11}, and
\textsc{Minimum Sum Edge Coloring}~\cite{DBLP:journals/dam/Marx09}
\footnote{We are grateful to Dr. Michael Lampis and Dr. Christian Komusiewicz for pointing to these papers~\cite{stackexchange}}.
Our reduction holds even when the larger parameter
`distance to paths' is two, 
i.e., when the input is restricted to graphs in which deleting two 
vertices result in a collection of paths.
Moreover, it also imply the following result.

\begin{theorem}
\label{thm:tree-depth-eth}
Unless the \ETH\ fails,
\textsc{Telephone Broadcast} does not admit
an algorithm running in time $2^{2^{o(\td)}} \cdot n^{\calO(1)}$,
where $\td$ is the tree-depth of the input graph.
\end{theorem}

Recently, Foucaud et al.~\cite{DBLP:journals/corr/abs-2307-08149}
proved that three distance-related \NP-\complete\
graph problems, viz. \textsc{Metric Dimension}, 
\textsc{Strong Metric Dimension}, \textsc{Geodetic Set}, 
admit such double exponential lower bounds
for parameters treewidth and vertex cover number.
Two follow up papers showed that problems from learning theory~\cite{DBLP:journals/corr/abs-2309-02876},
namely 
\textsc{Non-Clashing Teaching Map} and \textsc{Non-Clashing Teaching Dimension}, and local identification problems~\cite{DBLP:journals/corr/abs-2402-08346}, 
namely \textsc{Locating-Dominating Set} and \textsc{Test Cover} to admit similar lower bounds.

We organise the paper as follows:
Section~\ref{sec:prelims} contains necessary preliminaries and 
a simple but crucial observation (Observation~\ref{obs:message-propelling})
regarding \textsc{Telephone Broadcast}.
Section~\ref{sec:eth-lower-bound} contains proof of Theorem~\ref{thm:sol-size-lb} whereas Section~\ref{sec:np-hard-fvs} contains a reduction used
to prove both Theorem~\ref{thm:fvs-np-hard} and Theorem~\ref{thm:tree-depth-eth}.
We conclude the paper in Section~\ref{sec:conclusion} with an open problem.

\section{Preliminaries}
\label{sec:prelims}
For a positive integer $q$, we denote {the} set $\{1, 2, \dots, q\}$ by $[q]$.
We use $\mathbb{N}$ to denote the collection of all non-negative integers.

\subparagraph*{Parameterized complexity.} 
A \textit{parameterized problem} is a decision problem in which every 
instance $I$ is associated with a natural number 
$k$ called \textit{parameter}. 
A parameterized problem $\Pi$ is said to be 
\emph{fixed-parameter tractable} (\FPT) if every instance 
$(I,k)$ can be solved in $f(k)\cdot |I|^{\calO(1)}$ time where 
$f(\cdot)$ is some computable function whose value depends 
only on $k$. 
We say that two instances, $(I, k)$ and $(I’, k’)$, of a 
parameterized problem $\Pi$ are \emph{equivalent} 
if $(I, k) \in \Pi$ if and only if $(I’, k’) \in \Pi$. 
A parameterized problem $\Pi$ admits a kernel of size 
$g(k)$ (or $g(k)$-kernel) if there is a polynomial-time 
algorithm (called {\em kernelization algorithm}) 
which takes as an input $(I,k)$, and in time $(|I| + k)^{\calO(1)}$ returns
an equivalent instance $(I’,k’)$ of $\Pi$ such that $|I’| + k’ \leq g(k)$ 
where $g(\cdot)$ is a computable function whose value depends 
only on $k$. 

A parameterized problem $\Pi$ is said to be 
\emph{\para-\NP-\hard} parameterized by $k$ if 
the problem is \NP-\hard\ even for the constant value of 
parameter $k$.
The Exponential Time Hypothesis (\ETH), formalized in~\cite{DBLP:journals/jcss/ImpagliazzoP01},
roughly states that an arbitrary instance of $\psi$, with $n$ variables, of
\textsc{$3$-SAT} does not admit an algorithm running in time
$2^{o(n)}$.
For more details on parameterized complexity and \ETH, we refer the reader to the book by Cygan et al.~\cite{CyganFKLMPPS15}.

\subparagraph*{Graph Theory.} 
For an undirected graph $G$, sets $V(G)$ and $E(G)$ denote its set of vertices and edges, respectively.  
We denote an edge with two endpoints $u, v$ as $(u, v)$. 
Two vertices $u, v$ in $V(G)$ are \emph{adjacent} if there is an edge $(u, v)$ {in $G$}.
A \emph{simple path} from $v_1$ to $v_{d+1} (\neq v_1)$ is 
a non-empty sequence of vertices 
$\{v_1, v_{2}, \ldots, v_{d+1}\}$ such that all the 
vertices in this sequence are distinct and 
$(v_i, v_{i+1}) \in E(G)$ for every $i \in [d]$.
We define $\dist(v_1, v_{d+1})$ as the the number of edges
in the shortest path from $v_1$ to $v_{d+1}$.
A graph is called {\em connected} if there is a path between every pair of distinct vertices.
We refer to the book by Diestel~\cite{Diestel12} for standard graph-theoretic 
definitions and terminology not defined here.

A set of vertices $S$ of a graph $G$ is a \emph{vertex cover} if 
each edge of $G$ has at least one of its endpoints in $S$. 
The \emph{vertex cover number} of $G$ is the minimum size of a vertex cover. 
Note that for a vertex cover $S$, the set $I = V (G) \setminus S$ is 
an independent set, that is, any two distinct vertices of $I$ are not adjacent. 
A set of vertices $X$ of a graph $G$ is a \emph{feedback vertex set}
if $G - X$ is a forest, i.e., it does not contain any cycle. 
The \emph{feedback vertex set number} of $G$ is the minimum size of 
a feedback vertex set.
On similar lines, the \emph{distance to paths} of $G$ is
the minimum size of set $Y$ such that $G - Y$ is the collection of paths.
By definitions, the feedback vertex number of $G$ is less than or equal to its distance to paths, which is at most its vertex cover number.
We refer readers to~\cite[Chapter 7]{CyganFKLMPPS15}
for the definitions of pathwidth and treewidth.

Treedepth is a key invariant in the ‘sparsity theory’ for graphs 
initiated by Nesetril and Ossona de Mendez~\cite{DBLP:books/daglib/0030491}, with several algorithmic applications. 
It is defined as follows:
An \emph{elimination forest} of a graph $G$ is a rooted forest 
consisting of trees $T_1, \dots , T_p$ such that the sets
$V(T_1), \dots, V(T_p)$ partition the set $V(G)$ and for each edge 
$(x, y) \in E(G)$, the vertices $x$ and $y$ belong to one tree $T_i$ and 
in that tree, one of them is an ancestor of the other. 
The vertex-height of an elimination forest is the maximum number of vertices on a root-to-leaf path in any of its trees. 
The \emph{treedepth} of a graph $G$, denoted $\td(G)$, is the 
minimum vertex-height of an elimination forest of $G$.

\subparagraph*{Broadcasting in Graphs.} 
Let $G$ be a connected graph and let $s \in V (G)$ be a source vertex 
that broadcast a message. 
We say that a vertex is \emph{informed} if it has received
the messages and is \emph{uninformed} otherwise.
In general, a broadcasting protocol is a mapping that for each round 
$i \ge 1$, assigns to each vertex $v \in V (G)$ that is either a 
source or has received the message in rounds 
$1, \dots, i - 1$, a neighbor $u$ to which $v$ sends the message in 
the $i^{th}$ round.
We also refer to this as at time $i$, $u$ forwards the message to $v$. 
However, it is convenient to note that it can be assumed that each 
vertex $v$ that got the message, in the next $d \le \deg(v)$ rounds, 
transmits the message to some neighbors in a certain order in such 
a way that each vertex receives the message only once. 
This allows us to define a broadcasting protocol formally as a 
pair $(\calT,\{C(v) \mid v\in V(\calT)\})$, where 
$\calT$ is a spanning tree of $G$ rooted in $s$ and for each $v \in V(\calT)$, 
$C(v)$ is an ordered set of children of $v$ in $\calT$. 
As soon as $v$ gets the message, $v$ starts to send it to the children 
in $\calT$ in the order defined by $C(v)$. 

\begin{definition}[Message Propelling and Downtime]
Consider a path 
$P$ from $s$ to a vertex $v$.
We say the path \emph{propels the message} at time $t'$ if
there is an informed vertex $v_i$ who is the predecessor
of an uninformed vertex $v_{i+1}$ in the path,
and $v_i$ forwards the message to $v_{i+1}$. 
The \emph{downtime} of the path is the
number of occasions on which the 
path does not propel the message. 
\end{definition}

\begin{observation}
\label{obs:message-propelling}
Consider a \yes-instance $(G, s, t)$ of \textsc{Telephone Broadcast}.
Suppose there is a vertex $x \in V(G)$ such that $\dist(s, x) = t - \ell$
for a non-negative integer $\ell$.
Then, there is at least one path from $s$ to $x$ 
with downtime at most $\ell$.
\end{observation}
\begin{proof}
As $(G, s, t)$ is a \yes-instance of \textsc{Telephone Broadcast},
there exists a broadcasting protocol that can be represented
as a pair $(\calT,\{C(v) \mid v \in V(\calT)\})$, 
where $\calT$ is a spanning tree of 
$G$ rooted in $s$ and for each $v \in V(\calT)$, 
$C(v)$ is an ordered set of children of $v$ in $\calT$.
Suppose $x$ receives the message at time $t' \le t$.
Consider the unique path from $s$ to $x$ in $\calT$.
The downtime of this path is $t' - \dist_{\calT}(s, x)$,
 which is at most $t - \dist_{\calT}(s, x)$.
As $T$ is a spanning subtree of $G$, we have 
$\dist_G(s, x) \le \dist_{\calT}(s, x)$.
Hence, the downtime of the said path is at most
$t - \dist_G(s, x) = t - (t - \ell) = \ell$.
\end{proof}

\begin{figure}[t]
\centering
\includegraphics[scale=0.5]{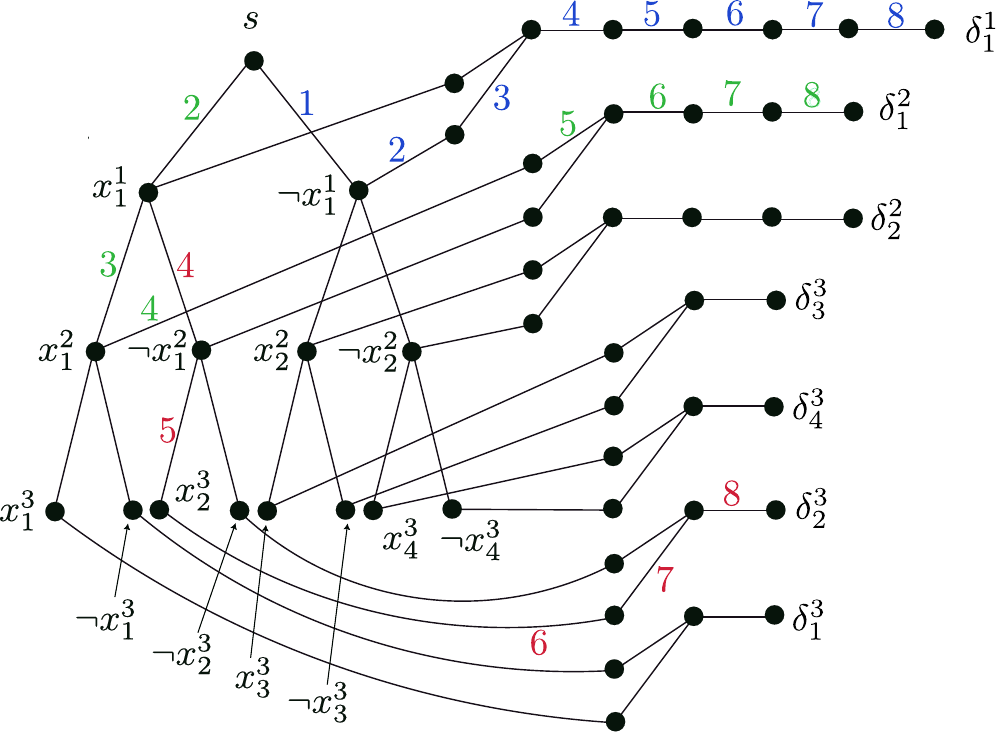}
\caption{An application of Observation~\ref{obs:message-propelling}.
For an edge $(u, v)$, the integer $t'$ associated with it denotes that
at time $t'$, vertex $u$ forwarded message to vertex $v$.
\label{fig:message-propelling}}
\end{figure}

Consider the toy example in Figure~\ref{fig:message-propelling} and
suppose $t = 8$.
Note that $\dist(s, \delta^1_1) = 8$, $\dist(s, \delta^2_1) = 7$, and
$\dist(s, \delta^3_3) = 6$.
There are two shortest paths $s$ to $\delta^1_1$,
one via $x^{1}_1$ and another via $\neg x^1_1$.
At least one of these paths should have a downtime of $0$.
Without loss of generality, suppose the path containing $\neg x^1_1$
has downtime $0$.
At every time step, the message is propelled along this path.
The blue time stamps in the figure denote this. 
Now consider the two shortest paths from $s$ to $\delta^2_1$,
one via $x^2_1$ and another via $\neg x^2_1$.
Both paths already have downtime of $1$ as $s$ forwarded
the message to $\neg x^1_1$.
Hence, by Observation~\ref{obs:message-propelling}, at least 
one of the paths has no more downtime.
Suppose such a path passes via $x^2_1$.
Hence, at every remaining time step, 
the message is propelled along this path, which 
is denoted by the green time stamps in the figure.
Finally, consider the shortest paths from $s$ to $\delta^3_2$,
one via $x^3_2$ and another via $\neg x^3_2$.
Both of these paths already occurred downtime of two
as $s$ forwarded the message to $\neg x^1_1$ and 
$x^1_1$ forwarded the message to $x^2_1$.
As before, this implies that at every remaining time step, 
the message is propelled along this path, which is 
denoted by the red time stamps in the figure.

\section{Lower Bound Parameterized by the Solution Size}
\label{sec:eth-lower-bound}

In this section, we prove Theorem~\ref{thm:sol-size-lb}, i.e.,
\textsc{Telephone Broadcast} does not admit
an algorithm running in time $2^{2^{o(t)}} \cdot |V(G)|^{\calO(1)}$,
unless the \ETH\ fails.
We present a reduction 
from a variant of \textsc{$3$-SAT} called
\textsc{$(3,3)$-SAT}.
In this variation, an input is a boolean satisfiability
formula $\psi$ in conjunctive normal form 
such that each clause contains \emph{at most}\footnote{We remark 
that if each clause contains \emph{exactly} $3$ variables,
and each variable appears $3$ times, then the problem
is polynomial-time solvable \cite[Theorem 2.4]{DBLP:journals/dam/Tovey84}.}
$3$ variables and each variable appears at most $3$ times.
Consider the following reduction from an instance $\phi$
of \textsc{$3$-SAT} with $n$ variables and $m$ clauses
to an instance $\psi$ of \textsc{$(3, 3)$-SAT} mentioned in 
\cite{DBLP:journals/dam/Tovey84}:
For every variable $x_i$ that appears $k~(> 3)$ times,
the reduction creates $k$ many new variables
$v^1_i, v^2_i, \dots, v^{k}_i$,
replaces the $j^{th}$ occurrence of $v_i$ by $v^{j}_i$, 
and adds a series of new clauses to encode
$v^1_i \Rightarrow v^2_i \Rightarrow \cdots \Rightarrow v^k_i 
\Rightarrow v^1_i$.
For an instance $\psi$ of \textsc{$3$-SAT},
suppose $k_i$ denotes the number of times a variable $x_i$
appeared in $\phi$.
Then, $\sum_{i \in [n]} k_i \le 3 \cdot m$.
Hence, the reduced instance
$\psi$ of \textsc{$(3,3)$-SAT} has at most $3m$ variables
and $4m$ clauses.
Using the \ETH~\cite{DBLP:journals/jcss/ImpagliazzoP01} and 
the sparsification lemma~\cite{DBLP:journals/jcss/ImpagliazzoPZ01}, 
we have the following result.

\begin{proposition}
\label{prop:3-3-SAT-ETH-lb}
\textsc{$(3,3)$-SAT}, with $n$ variables and $m$ clauses, 
does not admit an algorithm running in time $2^{o(m+n)}$,
unless the \emph{\ETH} fails.
\end{proposition}

We highlight that every variable appears positively and negatively
at most once.
Otherwise, if a variable appears only positively (respectively, only negatively)
then we can assign it \true\ (respectively, \false) 
and safely reduce the instance by removing the clauses 
containing this variable.
Hence, instead of using the terms 
the first, second, or third appearance of the variable,
we use the terms like first positive, first negative, second positive, or
second negative appearance of the variable. 

\subparagraph*{Bird's Eye View of the Reduction}
We start with a high-level overview of the reduction.
Consider an instance $\psi$ of \textsc{$(3, 3)$-SAT}
and the first four variables $v_1, v_2, v_3, v_4, \dots, $ and 
the first three clauses
$C_1 \equiv (\neg v_1 \lor \neg v_2 \lor v_4)$,
$C_2 \equiv (v_1 \lor v_2 \lor v_3)$,
$C_3 \equiv (v_1 \lor \neg v_2), \dots$~.
First, the reduction
partitions the variables into 
`buckets' $B^1 = \{v_1\}$,
$B^2 = \{v_2, v_3\}$, $\dots$~.
It rename the vertices $v_1, v_2, v_3, v_4, \dots$ to
$v^1_1, v^2_1, v^2_2, v^3_1, \dots$ to highlight 
the buckets these variables are in.
Then, it constructs a complete binary tree rooted
at the source $s$.
The $(\ell^{\circ})$$^{th}$ level of this tree corresponds to 
$(\ell^{\circ})$$^{th}$ bucket.
More precisely, for each variable $v^{\ell^{\circ}}_i$, 
it associates two vertices 
$x^{\ell^{\circ}}_i$ and $\neg x^{\ell^{\circ}}_i$
that have the same parent.
See Figure~\ref{fig:eth-lower-bound-overview}.

\begin{figure}[t]
\centering
\includegraphics[scale=0.5]{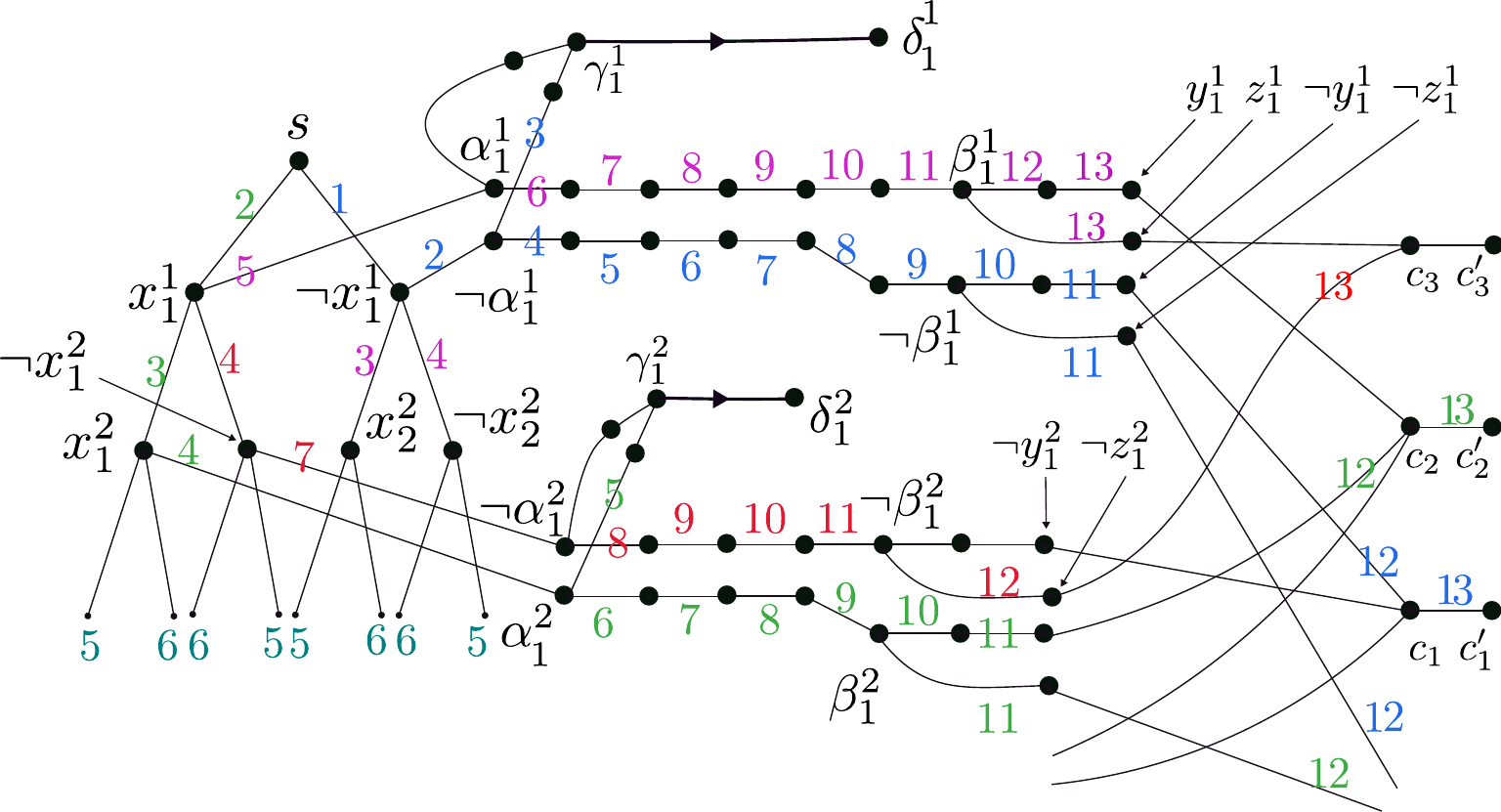}
\caption{Overview of the reduction in Section~\ref{sec:eth-lower-bound}.
For clarity, the figure only shows vertices corresponding to variables $v^1_1$ and $v^2_1$ (part of $v^2_2$ but not $v^3_1$).
With renaming the variables, $C_1 \equiv (\neg v^1_1 \lor \neg v^2_1 \lor v^3_1)$,
$C_2 \equiv (v^1_1 \lor v^2_1 \lor v^2_2)$, and $C_3 \equiv (v^1_1 \lor \neg v^2_1)$. 
$C_2$ and $C_3$ are connected to $y^1_1$ and $z^1_1$, respectively,
to encode the first positive and the second positive appearance of variable $v^1_1$.
The thick line with an arrow towards $\delta$-type vertices denotes the 
long paths whose length is adjusted so that Observation~\ref{obs:message-propelling} is applicable.
For this example, suppose $t = 13$.
The blue and green colored time-stamps show
a broadcasting protocol that forwards the messages to $\neg x^1_1$ 
(before forwarding it to $x^1_1$) and to 
$x^2_2$ (before forwarding it to $\neg x^2_2$), respectively.
This corresponds to assigning $v^1_1 = \false$ and $v^{2}_1 = \true$.
These choices ensure that vertices $c'_2$ and $c'_3$
get the message in time
which translate to satisfying clauses $C_2$ and $C_3$.
The red and purple colored time stamps, respectively, show 
that the message forwarded by $x^1_1$ or $\neg x^2_1$ 
can not reach the clause vertices.
In the case of red time stamps, the message can't reach
even if it deviates from others.
\label{fig:eth-lower-bound-overview}}
\end{figure}

In this complete binary tree, the order in which the
the message is forwarded from a parent to children, corresponding to the assignment of the variable.
For example, if $s$ forwards the message to $\neg x^1_1$
before $x^1_1$, this corresponds to
setting $v^1_1$ to $\false$.
Similarly, if $x^1_1$ (which is a parent of both 
$x^2_1$ and $\neg x^2_1$) forwards the message to $x^2_1$
before $\neg x^2_1$, this corresponds to
assigning $x^2_1$ to $\true$.
For each vertex of the form $x^{\ell^{\circ}}_i$, 
the reduction adds two vertices $y^{\ell^{\circ}}_i$, $z^{\ell^{\circ}}_i$ 
and uses them to connect to the clauses in which
variable $v^{\ell^{\circ}}_i$
appears positively for the first and the second 
time, respectively.
The vertex $x^{\ell^{\circ}}_i$ forwards the message towards
these vertices via the path containing $\alpha^{\ell^{\circ}}$ and 
$\beta^{\ell^{\circ}}$.
Note the way $\beta^{\ell^{\circ}}$ is connected to $y^{\ell^{\circ}}_i$, $z^{\ell^{\circ}}_i$ in Figure~\ref{fig:eth-lower-bound-overview}.
Hence, if the message is received at $\beta^{\ell^{\circ}}$ at time
$t'$, it can reach both $y^{\ell^{\circ}}_i$, $z^{\ell^{\circ}}_i$ at 
time $t' + 2$.
The reduction adds vertices corresponding to the negative
literal of the variable similarly.

Towards the clause side, the reduction adds an edge $(c_j, c'_j)$ 
for clause $C_j$ and makes $c_j$ adjacent with the 
\emph{$y$ or $z$-type} vertices representing the literal in it.
For every clause $C_j$,
the length of path connecting $\alpha^{\ell^{\circ}}$ to 
$\beta^{\ell^{\circ}}$ are adjusted such that 
the shortest path from $s$ to $\alpha^{\ell^{\circ}}$ to 
$\beta^{\ell^{\circ}}$ to $c_j$ is of length $t$ 
(if it is via $y^{\ell^{\circ}}$)
or $t - 1$ (if it is via $y^{\ell^{\circ}}$).
Hence, for each clause, at least one of \emph{$\beta$-type}
vertices representing the literal in it should receive 
the message in time $t - 4$.
This will translate to setting the literal to \true\ and, hence, satisfying
the clause.
For example, consider the clause $C_2 \equiv (v^1_1 \lor v^2_1 \lor v^2_2)$.
For the message to reach at $c'_2$ in time $t$, it should reach
at least one of $\beta^1_1$, $\beta^2_1$ or $\beta^2_2$ 
at time $t - 4$.
The reduction sets the lengths of paths such that if the message
is to reach $c_j$ via $\beta^2_1$, it should reach
$x^2_1$ before it reaches $\neg x^2_1$.
This translates into setting the variable $v^2_1$ to \true\ which
satisfies the clause $C_2$.

It remains to ensure that the message flows in the complete
binary tree in the desired way.
Towards this, for every variable $v^{\ell^{\circ}}_i$,
it adds a path with specified length from $\gamma^{\ell^{\circ}}_i$
to $\delta^{\ell^{\circ}}_i$ and 
connects $\gamma^{\ell^{\circ}}_i$ with a path of length two to both
$\alpha^{\ell^{\circ}}_i$ and $\neg \alpha^{\ell^{\circ}}_i$.
See Figure~\ref{fig:eth-lower-bound-overview}. 
The path of length two ensures that the message does not `jump'
from $(\neg \alpha^{1}_1)$-line to $(\alpha^{1}_1)$-line.
The length of this path depends on the bucket $B^{\ell^{\circ}}$ containing variable $v^{\ell^{\circ}}_i$ and adjusted such that
$\dist(s, \delta^{\ell^{\circ}}_i) = t - (\ell^{\circ} - 1)$. 
Hence, the shortest path from $s$ to $\delta^{\ell^{\circ}}_i$ can 
have a downtime of at most $\ell^{\circ} - 1$.
The reduction relies on these `distant vertices' and 
Observation~\ref{obs:message-propelling}, to ensure that
the broadcasting of the message corresponds to a valid
assignment of the variables.

\subparagraph*{Reduction}
The reduction takes as input an instance $\psi$ of 
\textsc{$(3,3)$-SAT} with $n$ variables and 
outputs an instance $(G,s,t)$ of \textsc{Telephone Broadcast}. 
The reduction modifies the given instance by renaming the variables and 
then constructs graph $G$ as mentioned below.

\begin{itemize}

\item Suppose $\calV = \{v_1, \dots, v_n\}$ is the collection 
of variables and $C = \{C_1, \dots, C_m\}$ is the collection
of clauses in $\psi$.
Here, we consider 
$\langle v_1, \dots, v_n\rangle$ and
$\langle C_1, \dots, C_m \rangle$ be arbitrary but fixed
orderings of variables and clauses in $\psi$.
The second ordering specifies the first/second positive/negative
appearance of variables in $\calV$ in the natural way. 

\item It groups the variables in the various buckets as follows:
\begin{itemize}
\item It processes the variables as per the ordering above.
It initializes the process by creating the first
bucket $B^{1} = \{v_1\}$ and assigning $\ell^{\circ} = 1$.
Until there is a variable that is not assigned a bucket,
it increments $\ell^{\circ}$ to $\ell^{\circ} + 1$,
creates an empty bucket $B^{\ell^{\circ}}$, and adds next
$2^{\ell^{\circ} - 1}$ many variables in the sequence to it.

\item Suppose the reduction has constructed $(\ell - 1)$ many
buckets so far. Note that for every $\ell^{\circ} < \ell - 1$,
bucket $B^{\ell^{\circ}}$ contains exactly $2^{\ell^{\circ} - 1}$ 
many variables.
Reduction add dummy variables to $\calV$ (towards the end 
of the above sequence), to ensure that even bucket 
$B^{\ell - 1}$ is full, i.e., it contains $2^{(\ell - 1) - 1}$ many
variables.
Finally, it adds another $2^{\ell - 1}$ many dummy variables
to $\calV$ (towards the end of the above sequence),
and constructs the bucket $B^{\ell}$ containing all these variables.
\item For every bucket $B^{\ell^{\circ}}$, 
it renames variables in it to 
$v^{\ell^{\circ}}_1, v^{\ell^{\circ}}_2, \dots, v^{\ell^{\circ}}_{2^{{\ell^{\circ}}-1}}$.
\end{itemize}

We remark that this last bucket of dummy variables is critical for the 
proof of correctness.
Moreover, the total number of variables now is at most
four times the number of variables in the original instance.
\item It starts constructing the graph $G$ by adding 
the source vertex $s$, and then adds a complete binary tree of 
height $\ell$ which is rooted at $s$.
We define \emph{level} of a vertex in the tree as its
distance from the root $s$.
Hence, for every $\ell^{\circ} \in [\ell]$, 
there are $2 \cdot 2^{\ell^{\circ} - 1}$ many vertices
in level $\ell^{\circ}$.
It renames the vertices of the tree as follows:
\begin{itemize}
\item It renames left child of $s$ as $x^1_1$ and 
right child of $s$ as $\neg x^1_1$.
\item For every $\ell^{\circ} \in \{1, 2, 3, \dots, \ell - 1\}$,
and $i \in [2^{\ell^{\circ} - 1}]$,
consider the vertex $x^{\ell^{\circ}}_i$ and $x^{\ell^{\circ}}_i$.
It renames two children of $x^{\ell^{\circ}}_i$ as
$x^{\ell^{\circ} + 1}_{2i - 1}$ and $\neg x^{\ell^{\circ} + 1}_{2i - 1}$,
and two children of $\neg x^{\ell^{\circ}}_i$ as
$x^{\ell^{\circ} + 1}_{2i}$ and $\neg x^{\ell^{\circ} + 1}_{2i}$.
\end{itemize}
Recall that for every $\ell^{\circ} \in [\ell]$,
bucket $B^{\ell^{\circ}}$ contains $2^{\ell^{\circ} - 1}$ variables,
viz $v^{\ell^{\circ}}_1, v^{\ell^{\circ}}_2, \dots, v^{\ell^{\circ}}_{2^{\ell^{\circ} - 1}}$.
For a variable $v^{\ell^{\circ}}_{i}$,
vertices $x^{\ell^{\circ}}_{i}$ and $\neg x^{\ell^{\circ}}_{i}$ 
correspond to positive and negative literal of $v^1_1$, respectively.
\item For every $\ell^{\circ} \in [\ell]$ and every $i \in [2^{\ell^{\circ} - 1}]$,
the reduction adds the following vertices and edges.
\begin{itemize}
\item It adds a path with end vertices  $\alpha^{\ell^{\circ}}_i$ and 
$\beta^{\ell^{\circ}}_i$ such that
\begin{align}
(2\ell^{\circ} - 1) + 1 + 1 + \dist(\alpha^{\ell^{\circ}}_i,
\beta^{\ell^{\circ}}_i ) &= t - 4.
\label{eq:alpah-beta-path}
\end{align}
We specify the value of $t$ at the end of the reduction.
Similarly, it adds another path with end vertices 
$\neg \alpha^{\ell^{\circ}}_i$ and $\neg \beta^{\ell^{\circ}}_i$
such that the distance between them satisfies the same condition.
\item It adds edges to make $\alpha^{\ell^{\circ}}_i$
adjacent with $x^{\ell^{\circ}}_i$
 and $\neg \alpha^{\ell^{\circ}}_i$ adjacent with $\neg x^{\ell^{\circ}}_i$.
\item It adds two vertices $y^{\ell^{\circ}}_i$ and $z^{\ell^{\circ}}_i$.
It adds a path of length two to connect $\beta^{\ell^{\circ}}_i$ to $y^{\ell^{\circ}}_i$, and adds an edge to make $\beta^{\ell^{\circ}}_i$
adjacent with $z^{\ell^{\circ}}_i$.
Similarly, it adds $\neg y^{\ell^{\circ}}_i$, $\neg z^{\ell^{\circ}}_i$,
and a path of length two to connect $\neg \beta^{\ell^{\circ}}_i$ to 
$\neg y^{\ell^{\circ}}_i$, and adds an edge to make $\beta^{\ell^{\circ}}_i$
adjacent with $\neg z^{\ell^{\circ}}_i$.
\item It adds a path with end vertices
$\gamma^{\ell^{\circ}}_i$ and $\delta^{\ell^{\circ}}_i$,
two paths of length two to connect $\gamma^{\ell^{\circ}}_i$
to $\alpha^{\ell^{\circ}}_i$ and $\neg \alpha^{\ell^{\circ}}_i$
such that
\begin{align}
\dist(s, \delta^{\ell^{\circ}}_i) = \dist(s, \alpha^{\ell^{\circ}}_i) + 
\dist(\alpha^{\ell^{\circ}}_i, \gamma^{\ell^{\circ}}_i) + 
 \dist(\gamma^{\ell^{\circ}}_i, \delta^{\ell^{\circ}}_i) &= t - (\ell^{\circ} - 1)\nonumber \\
 (\ell^{\circ}_i + 1) + 2 + \dist(\gamma^{\ell^{\circ}}_i, \delta^{\ell^{\circ}}_i) &= t - (\ell^{\circ} - 1).
 \label{eq:delta-dist}
 \end{align}
\end{itemize}
\item For every clause $C_j$, the reduction adds an edge $(c_j, c'_j)$.
Suppose $C_j$ contains the first positive (respectively, the second) 
appearance of variable $v^{\ell^{\circ}}_i$ for some $\ell^{\circ} \in [\ell]$ 
and $i \in [2^{\ell^{\circ} - 1}]$,
then, it adds an edge to make $y^{\ell^{\circ}}_i$ (respectively, 
$z^{\ell^{\circ}}_i$) adjacent with $c_j$.
Similarly, if $C_j$ contains the first positive (respectively, the second) 
appearance of variable $v^{\ell^{\circ}}_i$,
then, it adds an edge to make $\neg y^{\ell^{\circ}}_i$ (respectively, 
$\neg z^{\ell^{\circ}}_i$) adjacent with $c_j$.
\end{itemize}

This completes the construction.
The reduction sets $t = 2\ell + 6$ to ensure that all the distances 
mentioned above are at least two,  and returns $(G, s, t)$ 
as an instance of \textsc{Telephone Broadcast}.
In the following two lemmas, we prove the correctness of the reduction.

\begin{lemma}
\label{lemma:eth-lower-bound-forward}
If $\psi$ is a \yes-instance of \textsf{$(3,3)$-SAT} 
then $(G, s, t)$ is a \yes-instance of \textsc{Telephone Broadcast}.
\end{lemma}
\begin{proof}
Consider the following invariant:
For any $\ell^{\circ} \in [\ell]$ and 
$i \in [2^{\ell^{\circ} - 1}]$, consider the two vertices 
$x^{\ell^{\circ}}_i$ and $\neg x^{\ell^{\circ}}_i$ corresponding
to variable $v^{\ell^{\circ}}_i$.
The times by which these two vertices receive the message
are in $\{2\ell^{\circ} - 1, 2\ell^{\circ}\}$.
We define a broadcasting protocol that maintains this invariant. 
Suppose $\pi: \calV \mapsto \{\true, \false\}$ is 
a satisfying assignment for $\psi$.

\begin{itemize}
\item In the first round, i.e., at time $t' = 1$, 
if $\pi(v^1_1) = \true$ the source $s$ forwards the message
to $x^{1}_1$ otherwise it forwards it to $\neg x^{1}_1$.
\item For any $\ell^{\circ} \in \{2, 3, \dots, \ell\}$ and 
$i \in [2^{\ell^{\circ}}]$,
suppose $x^{\ell^{\circ}}_i$ has received the message 
before $\neg x^{\ell^{\circ}}_i$ and at time $t'$.
Then, $x^{\ell^{\circ}}_i$ forwards it to $\alpha^{\ell^{\circ}}_i$
at time $t' + 1$.
Note that by this time, $\neg x^{\ell^{\circ}}_i$ has received
the message from its parent in the complete binary tree.
In the case when $\neg x^{\ell^{\circ}}_i$ has received the message 
before $x^{\ell^{\circ}}_i$ at time $t'$, 
it forwards it to $\neg \alpha^{\ell^{\circ}}_i$ at time $t' + 1$.

\item Consider the time $t'+1$ when both $x^{\ell^{\circ}}_i$  and
$\neg x^{\ell^{\circ}}_i$ have received the message.
Recall that two children of $x^{\ell^{\circ}}_i$ are
$x^{\ell^{\circ+1}}_{2i-1}$ and $\neg x^{\ell^{\circ+1}}_{2i-1}$.
If $\pi(v^{\ell^{\circ+1}}_{2i-1}) = \true$ then $x^{\ell^{\circ}}_i$ 
forwards the message to $x^{\ell^{\circ+1}}_{2i-1}$ 
otherwise it forwards it to $\neg x^{\ell^{\circ+1}}_{2i-1}$.
Similarly, two children of $\neg x^{\ell^{\circ}}_i$ are
$x^{\ell^{\circ+1}}_{2i}$ and $\neg x^{\ell^{\circ+1}}_{2i}$.
If $\pi(v^{\ell^{\circ+1}}_{2i}) = \true$ then $\neg x^{\ell^{\circ}}_i$ 
forwards the message to $x^{\ell^{\circ+1}}_{2i}$ 
otherwise it forwards it to $\neg x^{\ell^{\circ+1}}_{2i}$.

\item Amongst $\alpha^{\ell^{\circ}}_i$ and $\neg \alpha^{\ell^{\circ}}_i$,
whomever receives the message first, forwards it towards 
$\gamma^{\ell^{\circ}}_i$, which continues to forward it 
in the direction of $\delta^{\ell^{\circ}}_i$.
Suppose $\gamma^{\ell^{\circ}}_i$ receives the message via
$\alpha^{\ell^{\circ}}_i$, then after forwarding the message in
the direction of $\delta^{\ell^{\circ}}_i$, it forwards it to
the vertex adjacent with both $\gamma^{\ell^{\circ}}_i$ and 
$\neg \alpha^{\ell^{\circ}}_i$.

\item Then, $\alpha^{\ell^{\circ}}_i$ (respectively, $\neg \alpha^{\ell^{\circ}}_i$), forwards the message towards 
$\beta^{\ell^{\circ}}_i$ (respectively, towards $\neg \beta^{\ell^{\circ}}_i$).
Whenever $\beta^{\ell^{\circ}}_i$ (respectively, $\neg \beta^{\ell^{\circ}}_i$) receives the message, it 
first forwards it to $y^{\ell^{\circ}}_i$ (respectively, to $\neg y^{\ell^{\circ}}_i$) 
and then to $z^{\ell^{\circ}}_i$ (respectively, to $\neg z^{\ell^{\circ}}_i$). 
 
 \item Whenever $y^{\ell^{\circ}}_i$ and $z^{\ell^{\circ}}_i$ (respectively, 
 $\neg y^{\ell^{\circ}}_i$ and $\neg z^{\ell^{\circ}}_i$) receive
 the message, they forward it to the unique clause (if any) vertex they are 
 adjacent them (if time permits).
\end{itemize}
This completes the protocol.

We now prove that every vertex receives the message in time $t$.
For any $\ell^{\circ} \in [\ell]$ and $i \in [2^{\ell^{\circ}-1}]$, 
consider the pair of vertices $x^{\ell^{\circ}}_i$ and $\neg x^{\ell^{\circ}}_i$.
From the invariant that the above protocol maintains, these two vertices
receive the message at time $2\ell^{\circ}$ and $2\ell^{\circ} - 1$
(which need not be in the same order).
As $t = 2\ell + 6$, every vertex in the complete binary
tree receives the message.
Before moving forward, consider the vertex in $\{x^{\ell^{\circ}}_i, \neg x^{\ell^{\circ}}_i\}$ that received the message at time $2\ell^{\circ} - 1$.
As $\dist(s, x^{\ell^{\circ}}_i) = \dist(s, \neg x^{\ell^{\circ}}_i) = \ell^{\circ}$,
this implies the path from $s$ to the vertex has downtime
of $\ell^{\circ}-1$.

Now consider the vertex of type $\delta^{\ell^{\circ}}_i$.
By the protocol, one of the paths from $s$ to $\delta^{\ell^{\circ}}_i$,
which are either via $x^{\ell^{\circ}}_i$ or $\neg x^{\ell^{\circ}}_i$, 
has downtime of $\ell^{\circ}-1$ till the message reaches 
one of these two vertices.
However, after this, the message is continuously propelled towards
$\delta^{\ell^{\circ}}_i$, and hence there is no more downtime
on this path.
As $\dist(s, \delta^{\ell^{\circ}}_i) = t - (\ell^{\circ} - 1)$,
the vertex $\delta^{\ell^{\circ}}_i$ receives the message 
in time $t$.
We note that since the $\dist(\gamma^{\ell^{\circ}}_{i}, 
\delta^{\ell^{\circ}}_{i})$ is at least two, the vertices
in the path connected $\alpha^{\ell^{\circ}}_{i}$ and 
$\neg \alpha^{\ell^{\circ}}_{i}$ to $\gamma^{\ell^{\circ}}_{i}$ 
also receive the message before it reaches $\delta^{\ell^{\circ}}_{i}$.
See Figure~\ref{fig:eth-lower-bound-layers}.

Now, consider the vertices of type $\beta^{\ell^{\circ}}_i$
and suppose $x^{\ell^{\circ}}_i$ receives the message before
$\neg x^{\ell^{\circ}}_i$.
Recall Equation~\ref{eq:alpah-beta-path}. 
The first term $(2\ell^{\circ}-1)$ corresponds to the time
by which $x^{\ell^{\circ}}_i$ receives the message.
The next term, i.e., $+1$ accounts for the round in which 
$x^{\ell^{\circ}}_i$ forwards the message towards $\beta^{\ell^{\circ}}_i$.
The second $+1$ term accounts for the round in which 
$\alpha^{\ell^{\circ}}_i$ forwards the message towards 
$\gamma^{\ell^{\circ}}_i$.
After that, the path continuously propers the message towards 
$\beta^{\ell^{\circ}}_i$.
Hence, by Equation~\ref{eq:alpah-beta-path}, $\beta^{\ell^{\circ}}_i$
receives the message at time $t - 4$ when 
$x^{\ell^{\circ}}_i$ receives the message before
$\neg x^{\ell^{\circ}}_i$, i.e., when $\pi(v^{\ell^{\circ}}_i) = \true$.
Using identical arguments, $\beta^{\ell^{\circ}}_i$
receives the message at time $t - 4$ when $\pi(v^{\ell^{\circ}}_i) = \false$.
We can extend this argument further to claim
that $y^{\ell^{\circ}}_i$ and $z^{\ell^{\circ}}_i$ receive the message
at time $t - 2$ if $\pi(v^{\ell^{\circ}}_i) = \true$.
Similarly, $\neg y^{\ell^{\circ}}_i$ and $\neg z^{\ell^{\circ}}_i$ 
receive the message at time $t - 2$ if $\pi(v^{\ell^{\circ}}_i) = \false$.

As $\pi$ is the satisfying assignment, for every clause
$C_j$, we can fix a literal that satisfies it.
By the arguments in the previous paragraph, 
the $y$-type or $z$-type vertex
corresponding to that literal has received the message at time $t - 2$.
Hence, the vertex $c'_j$ receives the message in time.
The only case that remains to be argued is the vertices on 
$\neg \alpha^{\ell^{\circ}}_i$ to $\neg \beta^{\ell^{\circ}}_i$ when 
$x^{\ell^{\circ}}_i$ receives the message before $x^{\ell^{\circ}}_i$.
Consider the vertices on paths 
$\alpha^{\ell^{\circ}}_i$ to $\beta^{\ell^{\circ}}_i$
and the corresponding vertices, i.e., vertices that are at the same distance
from $s$ on $\neg \alpha^{\ell^{\circ}}_i$ to $\neg \beta^{\ell^{\circ}}_i$.
For each pair of such vertices,
the time difference between the messages received is exactly two. 
Hence, in the last two rounds, i.e., while $y^{\ell^{\circ}}_i$ and $z^{\ell^{\circ}}_i$ are forwarding messages to vertices encoding the clause,
vertices $\neg y^{\ell^{\circ}}_i$ and $\neg z^{\ell^{\circ}}_i$ receives
the message.
This concludes that all the vertices in the graph receive the message in time.
\end{proof}

\begin{figure}[t]
\centering
\includegraphics[scale=0.65]{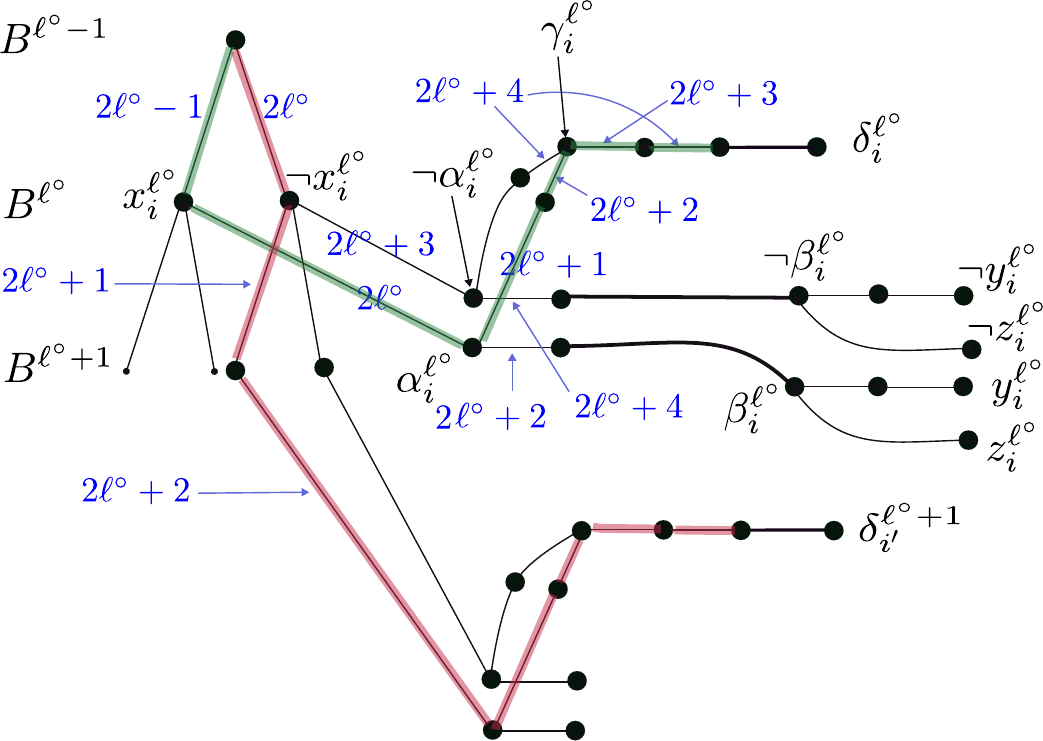}
\caption{Illustration of a pair of vertices in bucket $B^{\ell^{\circ}}$.
For the sake of clarity, we do not show all the vertices.
For Lemma~\ref{lemma:eth-lower-bound-forward}, 
the blue time stamps show the round in which a vertex forwards
the message along the edge.
For Lemma~\ref{lemma:eth-lower-bound-backward},
the green-shaded path and red-shaded path show the message propelled along the path due to constraints imposed by $\delta^{\ell^{\circ}}_i$
and $\delta^{\ell^{\circ} +1}_{i'}$, respectively.
\label{fig:eth-lower-bound-layers}}
\end{figure}

\begin{lemma}
\label{lemma:eth-lower-bound-backward}
If $(G, s, t)$ is a \yes-instance of \textsc{Telephone Broadcast}
then $\psi$ is a \yes-instance of \textsf{$(3,3)$-SAT}.
\end{lemma}
\begin{proof}
Consider a broadcasting protocol that is represented by
$(\calT,\{C(v) \mid v\in V(\calT)\})$, where $\calT$ is a spanning tree of $G$ 
rooted in $s$ and for each $v \in V(\calT)$, 
$C(v)$ is an ordered set of children of $v$ in $\calT$.
Recall that as soon as $v$ gets the message, $v$ starts to send it to the children 
in $\calT$ in the order defined by $C(v)$. 
We establish specific properties of this broadcasting.

\begin{claim}
\label{claim:x-forward-alpha}
For any $\ell^{\circ} \in [\ell]$ and 
$i \in [2^{\ell^{\circ} - 1}]$, consider the two vertices 
$x^{\ell^{\circ}}_i$ and $\neg x^{\ell^{\circ}}_i$ corresponding
to variable $v^{\ell^{\circ}}_i$.
Suppose $x^{\ell^{\circ}}_i$ receives the message at 
$2\ell^{\circ} - 1$ and $\neg x^{\ell^{\circ}}_i$ receives the message
after $x^{\ell^{\circ}}_i$.
Then, in round $2\ell^{\circ}$, vertex $x^{\ell^{\circ}}_i$ forwards 
the message to $\alpha^{\ell^{\circ}}_i$.
\end{claim}
\begin{claimproof}
Consider the vertex $\delta^{\ell^{\circ}}_i$.
By Equation~\ref{eq:delta-dist}, it is at distance $t - (\ell^{\circ} - 1)$ 
from $s$.
By Observation~\ref{obs:message-propelling}, there
is at least one path from $s$ to $\delta^{\ell^{\circ}}_i$ 
that has at most $(\ell^{\circ} - 1)$ downtime.
Once again, there are only two paths of length $t - (\ell^{\circ} - 1)$
from $s$ to $\delta^{\ell^{\circ}}_i$, viz one via $x^{\ell}_i$ and
another via $\neg x^{\ell}_i$.
By the construction, $\dist(s, x^{\ell^{\circ}}_i) = \ell^{\circ}$,
and hence the path from $s$ to $x^{\ell^{\circ}}_i$ and
the path from $s$ to $\neg x^{\ell^{\circ}}_i$ has downtime
of $\ell^{\circ} - 1$ and at least $\ell^{\circ}$, respectively.
This implies that the path from $x^{\ell^{\circ}}_i$ to 
$\delta^{\ell^{\circ}}_i$ can not have any more downtime.
Hence, at time $2\ell^{\circ}$, vertex $x^{\ell^{\circ}}_i$ forwards 
the message to $\alpha^{\ell^{\circ}}_i$.
\end{claimproof}

\begin{claim}
\label{claim:broacast-to-asst}
For any $\ell^{\circ} \in [\ell - 1]$ and 
$i \in [2^{\ell^{\circ} - 1}]$, consider the two vertices 
$x^{\ell^{\circ}}_i$ and $\neg x^{\ell^{\circ}}_i$ corresponding
to variable $v^{\ell^{\circ}}_i$.
The times by which these two vertices receive the message
are $2\ell^{\circ} - 1$ and $2\ell^{\circ}$ (but need not be
in the same order).
\end{claim}
\begin{claimproof}
We prove this claim using the induction on $\ell^{\circ}$.
Consider the base case when $\ell^{\circ} = 1$.
By Equation~\ref{eq:delta-dist}, $\delta^{\ell^{\circ}}_i$
is at distance $t - (\ell^{\circ} - 1)$.
By Observation~\ref{obs:message-propelling}, at least one path from $s$ to $\delta^1_1$ has zero downtime.
Note that there are only two paths from $s$ to $\delta^1_1$ of distance
exactly $t$,
viz one passing through $x^1_1$ and another passing through
$\neg x^{1}_1$.
Hence, at least one of these two paths has zero downtime.
Without loss of generality, suppose the path via $\neg x^{1}_1$
has a downtime of zero.
Hence, $s$ forwards the message to $\neg x^{1}_1$ in the first round.
As $s$ has only two children, $s$ forwards the message to $x^1_1$ 
in the second round.
This completes the base case.

Suppose the induction hypothesis states the claim is true
for $\ell^{\circ} - 1$.
We prove the claim holds for $\ell^{\circ}$.
Consider a pair of vertices $x^{\ell^{\circ}}_i$ and $\neg x^{\ell^{\circ}}_i$ for
some $i \in [2^{\ell^{\circ} - 1}]$.
We consider the following two cases,
which are exhaustive by the induction hypothesis.

\emph{Case 1: The parent of these two vertices receives the message
at time $2(\ell^{\circ} - 1) - 1 = 2\ell^{\circ} - 3$.}

\emph{Case 2: The parent of these two vertices receives the message
at time $2(\ell^{\circ} - 1) = 2\ell^{\circ} - 2$.}

In the first case, 
by Claim~\ref{claim:x-forward-alpha}, the parent of $\ell^{\circ}_i$
forwards the message to 
$\alpha$-type vertex adjacent to it in round $2(\ell^{\circ} - 1)$.
Hence, in either case, the parent of 
$x^{\ell^{\circ}}_i$ and $\neg x^{\ell^{\circ}}_i$ can forward 
the message to its children in round $(2\ell^{\circ} - 2) + 1 = 2\ell^{\circ} - 1$.

Now, consider the vertex $\delta^{\ell^{\circ}}_i$.
Once again, there are only two paths of length $t - (\ell^{\circ} - 1)$
from $s$ to $\delta^{\ell^{\circ}}_i$, viz one via $x^{\ell}_i$ and
another via $\neg x^{\ell^{\circ}}_i$, but both 
contain the parent of these two vertices.
As the parent of these vertices can only forward the message 
in round $2\ell^{\circ} - 1$, 
and $\dist(s, x^{\ell^{\circ}}_i) = \dist(s, \neg x^{\ell^{\circ}}_i) = \ell^{\circ}$, there is already a downtime of $\ell^{\circ} - 1$ 
by the time message reaches $x^{\ell^{\circ}}_i$ or $\neg x^{\ell^{\circ}}_i$.
By Observation~\ref{obs:message-propelling}, there
is at least one path from $s$ to $\delta^{\ell^{\circ}}_i$ 
that has at most $\ell^{\circ} - 1$ downtime.
Hence, there is at least one path from the parent of 
$x^{\ell^{\circ}}_{i}$ and $\neg x^{\ell^{\circ}}_{i}$
to  $\delta^{\ell^{\circ}}_i$ without anymore downtime.
This implies that at time $2\ell^{\circ} - 1$, 
the parent of $x^{\ell^{\circ}}_{i}$ and $\neg x^{\ell^{\circ}}_{i}$
forwards the message to either of these two vertices.
Without loss of generality, suppose $x^{\ell^{\circ}}_{i}$ receives
the message by time $2\ell^{\circ} - 1$, and hence 
before $\neg x^{\ell^{\circ}}_{i}$.
See the green shaded path in Figure~\ref{fig:eth-lower-bound-layers}.
It remains to argue that the $\neg x^{\ell^{\circ}}_{i}$ receives the 
message in time $2\ell^{\circ}$, i.e., in the next round.

Suppose $x^{\ell^{\circ} +1}_{i'}$ and $\neg x^{\ell^{\circ} +1}_{i'}$
are the children of $\neg x^{\ell^{\circ}}_{i}$.
Now, consider the vertex $\delta^{\ell^{\circ} +1}_{i'}$.
Once again, there are only two paths of length $t - (\ell^{\circ} + 1 - 1)$
from $s$ to $\delta^{\ell^{\circ} + 1}_i$, viz one via $x^{\ell^{\circ} + 1}_{i'}$ and another via $\neg x^{\ell^{\circ} + 1}_{i'}$, 
but both contain parent of $\neg x^{\ell^{\circ}}_i$.
As the parent of $\neg x^{\ell^{\circ}}_i$ can only forward the message 
in round $2\ell^{\circ}$, 
and $\dist(s, x^{\ell^{\circ}}_i) = \dist(s, \neg x^{\ell^{\circ}}_i) = \ell^{\circ}$, there is already a downtime of $\ell^{\circ}$ 
by the time message reaches $\neg x^{\ell^{\circ}}_i$.
By Observation~\ref{obs:message-propelling}, there
is at least one path from $s$ to $\delta^{\ell^{\circ}}_i$ 
that has at most $\ell^{\circ}$ downtime.
Hence, there is at least one path from the parent of 
$\neg x^{\ell^{\circ}}_{i}$
to  $\delta^{\ell^{\circ}+1}_{i'}$ without anymore downtime.
This implies that at time $2\ell^{\circ}$, 
the parent of $\neg x^{\ell^{\circ}}_{i}$
forwards the message to $\neg x^{\ell^{\circ}}_{i}$.
See the red shaded path in Figure~\ref{fig:eth-lower-bound-layers}.
This concludes the proof of the claim.
\end{claimproof}

We use the above claim to construct a satisfying assignment
for $\psi$. 
The above claim is valid only for $\ell^{\circ} < \ell$.
However, the last bucket contains dummy variables, and hence, such a claim is not necessary. 

Consider the following assignment $\pi: \calV \mapsto \{\true, \false\}$.
For every $\ell^{\circ} \in [\ell - 1]$ and 
$i \in [2^{\ell^{\circ}} - 1]$, assign $\pi(v^{\ell^{\circ}}_i) = \true$
if $x^{\ell^{\circ}}_i$ receives the message in time $(2\ell^{\circ} - 1)$
otherwise, assign it to \false.
For the variables in $B^{\ell}$, arbitrarily assign them to \true\ or \false.
Claim~\ref{claim:broacast-to-asst} ensures that this is a valid assignment.
In the remaining proof, we argue that it is a satisfying assignment.
Before that, note that the proof of 
Claim~\ref{claim:broacast-to-asst} implies that
if $\pi(v^{\ell^{\circ}}_i) = \true$ then
$x^{\ell^{\circ}}_i$ receives the message in time $2\ell^{\circ} - 1$
and hence $\alpha^{\ell^{\circ}}_i$ receives the message in time
$2\ell^{\circ}$.
Similarly, if $\pi(v^{\ell^{\circ}}_i) = \false$ then
$\neg x^{\ell^{\circ}}_i$ receives the message in time $2\ell^{\circ} - 1$
and hence $\neg \alpha^{\ell^{\circ}}_i$ receives the message in $2\ell^{\circ}$.
Moreover, whichever vertex amongst $\alpha^{\ell^{\circ}}_i$ or $\neg \alpha^{\ell^{\circ}}_i$ receives the message first
needs to send it in the direction of $\gamma^{\ell^{\circ}}_i$.
By Equation~\ref{eq:alpah-beta-path},
$\dist(\alpha^{\ell^{\circ}}_i, \beta^{\ell^{\circ}}_i) = 
\dist(\neg \alpha^{\ell^{\circ}}_i, \neg \beta^{\ell^{\circ}}_i) = (t - 4) - (2\ell^{\circ} + 1)$.
If $\pi(v^{\ell^{\circ}}_i) = \true$ then $\beta^{\ell^{\circ}}_i$ receives the message in time $t - 4$, else 
$\pi(v^{\ell^{\circ}}_i) = \false$ and $\neg \beta^{\ell^{\circ}}_i$ receives the message in time $t - 4$.

Now, consider an arbitrary clause $C_j$ of $\psi$ and the corresponding
vertices $c'_j$ in the graph.
We started with a broadcast protocol that ensures the message
reaches all the vertices, especially $c'$, 
at least one of the three $\beta$-type vertices corresponding to
literals in $C_j$ did not receive the message in time $t-4$.
This implies that for every clause $C_j$, at least one of its literal is 
set to \true, and hence  $\pi$ is a satisfying assignment which concludes 
the proof.
\end{proof}

\begin{proof}[Proof of Theorem~\ref{thm:sol-size-lb}]
Assume there is an algorithm $\calA$ that, given an 
instance  $(G, s, t)$ of \textsc{Telephone Broadcast}, 
runs in time $2^{2^{o(t)}} \cdot n^{\calO(1)}$
and correctly determines whether it is \yes-instance.
Consider the following algorithm that takes as input
an instance $\psi$ of \textsc{$(3, 3)$-SAT} and determines
whether it is a \yes-instance.
It first constructs an equivalent instance $(G, s, t)$ of 
\textsc{Telephone Broadcast} as mentioned in this section.
Then, it calls algorithm $\calA$ as a subroutine and returns 
the same answer.
The correctness of this algorithm follows from the correctness
of algorithm $\calA$, Lemma~\ref{lemma:eth-lower-bound-forward} 
and Lemma~\ref{lemma:eth-lower-bound-backward}.
By the description of the construction,
it takes time polynomial in the number of variables in $\psi$
to return an instance of $(G, s, t)$ of \textsc{Telephone Broadcast},
and $t = \calO(\log(n))$.
This implies the running time of the algorithm for 
\textsc{$(3, 3)$-SAT} is $2^{2^{o(\log (n))}} \cdot n^{\calO(1)} = 2^{o(n)} \cdot n^{\calO(1)}$.
This, however, contradicts Proposition~\ref{prop:3-3-SAT-ETH-lb}.
Hence, our assumption is wrong, and
\textsc{Telephone Broadcast} does not admit an algorithm 
running in time $2^{2^{o(t)}} \cdot |V(G)|^{\calO(1)}$, 
unless the {\ETH} fails.
\end{proof}


\section{On Graphs of Bounded Feedback Vertex Set}
\label{sec:np-hard-fvs}

In this section, we prove the problem in \NP-\complete\
even when restricted to graphs with the feedback vertex 
set number one.
The starting point of our reduction is the following problem.

\defproblem{\textsc{Numerical $3$-Dimensional Matching}}{Disjoint sets $W$, $X$, and $Y$, each containing $m$ elements,
a function $\size: W\cup X\cup Y \mapsto \mathbb{N}^{+}$, and a positive
integer $T$ such that $\sum_{a \in W\cup X\cup Y} \size(a) = m \cdot T$.}{Does there exist a partition of $W\cup X\cup Y$  into $m$
disjoint sets $A_1, A_2, \dots, A_m$ such that for every $i \in [m]$,
$A_i$ contains exactly one element from each of $W$, $X$, and $Y$,
and $\sum_{a\in A_i} \size(a) = T$?}

This problem is known to be strongly \NP-\complete\ (See [SP16] in \cite{DBLP:books/fm/GareyJ79}).
Specifically, by Theorem 4.4 in \cite{DBLP:books/fm/GareyJ79},
the problem is \NP-\complete\ even when $\max_{a \in W\cup X\cup Y} \{\size(a)\} \le 2^{16} \cdot (3m)^4 = c_0 \cdot m^4$ for a constant $c_0$.
We define a closely related problem that is more suitable for our
problem, and prove that this problem is alos strongly \NP-\complete.

\defproblem{\textsc{Numerical $3$-Dimensional (Almost) Matching}}{
Disjoint sets $W$, $X$, and $Y$, each containing $m$ elements,
a function $\size: W\cup X\cup Y \mapsto \mathbb{N}^{+}$ such that
$(i)$ for any two elements $x \neq x'$ in $X$, $\size(x) \neq \size(x')$,
$(ii)$ for any two elements $y \neq y'$ in $Y$, $\size(y) \neq \size(y')$,
$(iii)$ for any $a \in W \cup X\cup Y$, $\size(a)$ is multiplicative of $m^2$, 
and
 two positive integers $T, \lambda$.}{Does there exist a partition of $W\cup X\cup Y$  into $m$
disjoint sets $A_1, A_2, \dots, A_m$ such that for every $i \in [m]$,
$A_i$ contains exactly one element from each of $W$, $X$, and $Y$,
and $T - \lambda \le  \sum_{a\in A_i} \size(a)$?}

\begin{lemma}
\textsc{Numerical $3$-Dimensional (Almost) Matching} is \NP-\complete\
in the strong sense, i.e., even when $\max_{a \in W\cup X\cup Y} \{\size(a)\} \le 2^{16} \cdot (3m)^4 \cdot m^2$.
\end{lemma}
\begin{proof}
It is easy to verify that the problem is in \NP.
Consider a polynomial time reduction that takes as input an instance
$(W, X, Y, \size, T)$ of \textsc{Numerical $3$-Dimensional Matching}
and constructs an instance $(W, X, Y, \size^\circ, T^\circ, \lambda^\circ)$ 
of \textsc{Numerical $3$-Dimensional (Almost) Matching}
as mentioned below.
It is safe to assume that the number of elements $m$ in each set
is large, hence $m^2 \gg 2(m + 1)$.
\begin{itemize}
\item For every $a \in W \cup X \cup Z$,
define $\size^{\circ}(a) = m^2 \cdot \size(a)$.
Also, define $T^{\circ} = m^2 \cdot T$ and 
$\lambda = 2(m + 1)$. 
\item 
Suppose $(x_1, x_2,\dots, x_m)$ is the non-decreasing ordering of 
elements in $X$ with respect to $\size^{\circ}$ function. 
By the construction, for any $i \in [m - 1]$,
if $\size^{\circ}(x_i) \neq \size^{\circ}(x_{i+1})$
then $m^2 \le \size^{\circ}(x_{i+1}) - \size^{\circ}(x_{i})$.

\item Suppose $i$ is the smallest integer and $k \ge 2$ is the largest integer 
such that $x_{i+1}, x_{i+2}, \dots, x_{i + k}$ have the same size.
Then, redefine $\size^{\circ}(x_{i + j}) = \size^{\circ}(x_{i + j}) - (k - j)$.
As $k \le m$, we have
$\size^{\circ}(x_i) < \size^{\circ}(x_{i+1}) < \dots < \size^{\circ}(x_{i+k})$.
Hence, all the elements before $x_{i + k}$ have different sizes.
Repeating this process for all the elements in $X$ and then 
for elements in $Y$ results in the instance with the desired property.
\end{itemize}

It remains to prove that this new instance is equivalent to the input instance.
We prove that disjoint sets $A_1, A_2, \dots, A_m$ is a
solution for \textsc{Numerical $3$-Dimensional Matching}
if and only if it is a solution for 
\textsc{Numerical $3$-Dimensional (Almost) Matching}.
Consider an arbitrary set $A_i = \{w, x, y\}$
where $w \in W$, $x \in X$, and $y \in Y$.
It is sufficient to prove that  
$\size(w) + \size(x) + \size(y) = T$ if and only if
$T^{\circ} - \lambda \le \size^{\circ} (w) + \size^{\circ} (x) + \size^{\circ} (y)$.
The forward direction follows as the size of elements in $X \cup Y$ 
can reduce by at most $m$ and $\lambda = 2(m + 1)$.
In the backward direction, 
by the construction,
\begin{align*}
T^{\circ} - \lambda & \le \size^{\circ} (w) + \size^{\circ} (x) + \size^{\circ} (y)\\
m^2 \cdot T - \lambda & \le m^2 \cdot \size(w) + (m^2 \cdot \size(x) -\lambda_x) + (m^2 \cdot \size(y) - \lambda_y)\\
m^2 \cdot T - (\lambda - \lambda_x - \lambda_y)  & \le m^2 \cdot (\size(w) + \size(x) + \size(y)) \\
T - (\lambda - \lambda_x - \lambda_y)/m^2  &\le (\size(w) + \size(x) + \size(y)).
\end{align*}
Here, $0 \le \lambda_x, \lambda_y \le m$ are possible reduction in 
$\size^{\circ}(x)$ and $\size^{\circ}(y)$ because of multiple elements
having same weights.
As $\lambda = 2(m+1)$ and $m$ is large, we have $(\lambda - \lambda_x - \lambda_y)/m^2 \ll 1$.
As  $(\size(w) + \size(x) + \size(y))$ and $T$ are integers, we conclude that
$T \le (\size(w) + \size(x) + \size(y))$.
Recall that $A_i = \{w, x, y\}$ is an arbitrary set from the solution.
This implies that for an arbitrary set $A_i$, we have 
$T \le (\size(w) + \size(x) + \size(y))$.
As $\sum_{a \in W\cup X\cup Y} \size(a) = m \cdot T$, 
for every set $A_i$, 
we have $T = (\size(w) + \size(x) + \size(y))$.
This concludes the proof of the lemma.
\end{proof}

\begin{figure}[t]
\centering
\includegraphics[scale=0.65]{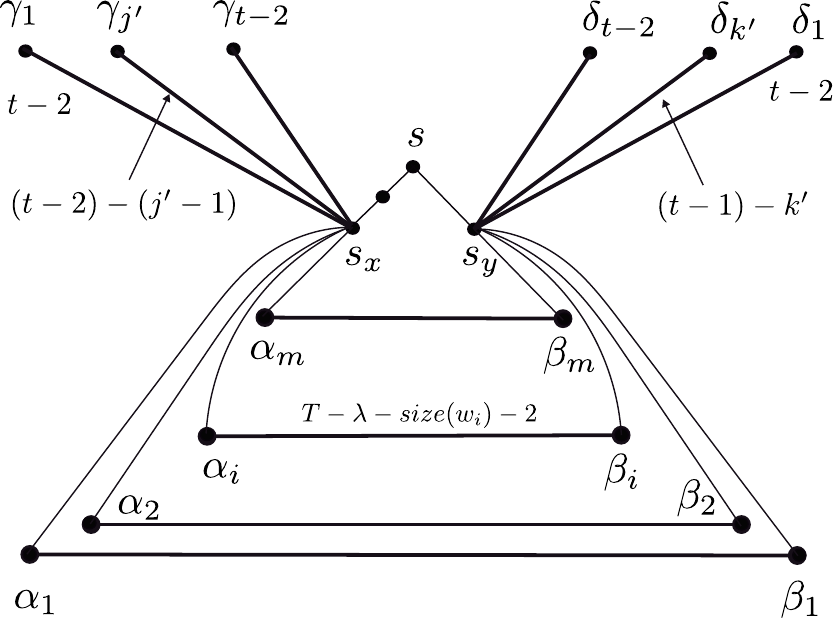}
\caption{Overview of the reduction. 
The thick line shows long paths of the specified size.
\label{fig:fvs-np-hard}}
\end{figure}

Consider the following reduction that takes as 
input an instance $(W, X, Y, \size, T, \lambda)$ of 
\textsc{Numerical $3$-Dimensional (Almost) Matching},
with the property that $\max_{a \in W\cup X\cup Y} \{\size(a)\} \le 2^{16} \cdot (3m)^4 \cdot m^2$,
and 
returns an instance $(G, s, t)$ of \textsc{Telephone Broadcast}.
Suppose $(w_1, w_2, \dots, w_m)$, $(x_1, x_2, \dots, x_m)$,
and $(y_1, y_2, \dots, y_m)$ be the arrangement of
elements in $W$, $X$, and $Y$, respectively accordingly to 
$\size$ function.
Note that the first sequence is non-decreasing, whereas
the other two are strictly increasing.
\begin{itemize}
\item 
For every element $w_i$ in $W$, the reduction adds two nodes
$\alpha_i$ and $\beta_i$ and connects them with the path of 
length $T - \lambda - \size(w_i) - 2$.
\item The reduction adds two vertices $s_x$ and $s_y$.
It adds edges to make $s_x$ adjacent with all vertices
in $\{\alpha_i \mid i \in [m]\}$ and $\{\beta_i \mid i \in [m]\}$.
\item For every $j' \in \{1, 2, \dots, t-2\} \setminus 
\{ t - \size(x_j) \mid j \in [m]\}$,
the reduction adds a vertex $\gamma_{j'}$ and connects it
to $s_x$ via path of length $(t - 2) - (j' - 1)$.

\item For every $k' \in \{1, 2, \dots, t - 2\} \setminus 
\{ t - \size(y_k) \mid k \in [m]\}$, 
the reduction adds a vertex $\delta_{k'}$ and connects it
to $s_y$ via path of length $(t - 1) - k'$.
\item Finally, it adds a source vertex $s$, connects it with
$s_x$ via path of length two and makes it adjacent with $s_y$.
\end{itemize}

This completes the construction of $G$.
The reduction returns $(G, s, t = T)$ as an instance of 
\textsc{Telephone Broadcast}.

We present a brief overview of the reduction.
The core idea is: $\gamma$-type vertices and $\delta$-type vertices
and too many and too far that they demand $s_x$ and $s_y$, respectively,
to keep forwarding the message in their directions.
The only time these two vertices are allowed to forward
message to $\alpha$-type vertices and $\beta$-type vertices are
at time $(t-\size(x_j))$ and $(t - \size(y_k))$ for some $x_j \in X$
and $y_k \in Y$.
Now, suppose vertex $\alpha_i$ gets the message at time $(t-\size(x_j))$
and $\beta_i$ gets the message at time $(t - \size(y_k))$.
These two vertices forward the message to each other 
to convey it to all the vertices in the path connecting them.
All the vertices in the path will be informed if and only if
$T - \lambda - \size(w_i) - 2 \le  (\size(x_j) - 1) + (\size(y_k) - 1)$ 
which implies 
$T - \lambda \le \size(w_i) + \size(x_j) + \size(y_k)$.
As this is true for every $i \in [m]$, 
there is a natural correspondence between 
a solution of \textsc{Telephone Broadcasting} and that of
\textsc{Numerical $3$-Dimensional (Almost) Matching}.
Note that selecting appropriate $x_j$ and $y_k$
for $w_i$ to make a partition depends on the time $s_x$ 
and $s_y$ forward message to $\alpha_i$ and $\beta_i$, respectively.
As any vertices can forward the message to at most
one of their neighbors in a round, we critically need that 
$\size(x_j) \neq \size(x_{j'})$ and $\size(y_k) \neq \size(y_{k'})$ 
for every $x_j \neq x_{j'} \in X$ and $y_k \neq y_{k'} \in Y$.
We formalize this intuition in the next two lemmas.

\begin{lemma}
\label{lemma:bounded-fvs-np-hard-forword}
If $(W, X, Y, \size, T, \lambda)$ is a \yes-instance of 
\textsc{Numerical $3$-Dimensional (Almost) Matching} then 
$(G, s, t)$ is a \yes-instance of 
\textsc{Telephone Broadcast}.
\end{lemma}
\begin{proof}
Suppose $A_1, A_2, \dots, A_m$ is a partition of \textsc{Numerical $3$-Dimensional (Almost) Matching} with desired properties.
Without loss of generality, suppose that $A_i$ contains $w_i$ for 
every $i \in [m]$.
Consider the following broadcasting protocol based on this partition.
\begin{itemize}
\item In the first round, $s$ forwards the message towards $s_x$.
In the second round, $s$ forwards the message to $s_y$, and the
vertex adjacent to $s$ and $s_x$ forwards it to 
$s_x$.
Hence, by the end of the second round, both $s_x$ and $s_y$
has received the message.
\item Consider vertex $s_x$ and subsequent rounds, i.e., when time 
$t' \ge 3$.
\begin{itemize}
\item If $t' - 2 = j' \in \{1, 2, \dots, t-2\} \setminus 
\{t - \size(x_j) \mid j \in [m]\}$ then
the $s_x$ forwards the message towards $\gamma_{j'}$.
In the subsequent rounds, the path from $s_x$ to $\gamma_{j'}$
keeps propelling the message towards $\gamma_{j'}$.
\item Suppose $t' - 2 = j' = t - \size(x_j)$.
As $A_1, A_2, \dots, A_m$ is a partition of $W \cup X \cup Y$
such that each set contains exactly one element from $W$, $X$, and $Y$,
there exists $A_i$  for some $i \in [m]$ such that $x_j \in A_i$.
Then, $s_x$ forwards message to $\alpha_i$.
In the subsequent rounds, the path from $\alpha_i$ to $\beta_i$
keeps propelling the message towards $\beta_i$.
\end{itemize}
\item The broadcasting protocol is defined in a similar way
at $s_y$. 
We repeat it for the sake of completeness.
\begin{itemize}
\item If $(t' - 2) = k' \in \{1, 2, \dots, t-2\} \setminus 
\{t - \size(y_k) \mid k \in [m]\}$ then
the $s_y$ forwards the message towards $\delta_{k'}$.
In the subsequent rounds, the path from $s_y$ to $\delta_{k'}$
keeps propelling the message towards $\delta_{k'}$.
\item Suppose $(t' - 2) = k' = t - \size(y_k)$.
As $A_1, A_2, \dots, A_m$ is a partition of $W \cup X \cup Y$
such that each set contains exactly one element from $W$, $X$, and $Y$,
there exists $A_i$  for some $i \in [m]$ such that $y_k \in A_i$.
Then, $s_y$ forwards message to $\beta_i$.
In the subsequent rounds, the path from $\beta_i$ to $\alpha_i$
keeps propelling the message towards $\alpha_i$.
\end{itemize}
\end{itemize}
This completes the protocol.

We now prove that every vertex in the graph gets the message
in time $T$.
As mentioned before, by the end of the second round, both $s_x$ and $s_y$
has received the message.
Consider the vertex $s_x$ in the subsequent round, i.e., at time $t' \ge 3$.
Define $t' - 2 = j'$, and suppose $j' \in \{1, 2, \dots, t-2\} \setminus 
\{t - \size(x_j) \mid j \in [m]\}$.
Then, by the protocol, $s_x$ forwards the message 
towards $\gamma_{j'}$ in the $(t')^{th}$ round. 
Hence, by the end of $(t')^{th}$ round, the message has reached the vertex
at a distance $1$ from $s_x$ on the path towards $\gamma_{j'}$.
This implies that by the end of $t^{th}$ round, the message has reached
the vertex at distance $t - (t' - 1) = t - (j' + 2 - 1) = (t - 2) - (j' - 1)$.
This implies that all the vertices in the path from $s_x$ to $\gamma_{j'}$
has received the message.

Similarly, consider the vertex $s_y$ in the subsequent round, i.e., at time $t' \ge 3$.
Define $t' - 2 = k'$, and suppose $k' \in \{1, 2, \dots, t-2\} \setminus 
\{t - \size(y_k) \mid k \in [m]\}$.
Then, by the protocol, $s_y$ forwards the message 
towards $\delta_{k'}$ in the $(t')^{th}$ round. 
Hence, by the end of $(t')^{th}$ round, the message has reached the vertex
at a distance $1$ from $s_x$ on the path towards $\gamma_{j'}$.
This implies that by the end of $t^{th}$ round, the message has reached
vertex at distance $t - (t' - 1) = t - (k' + 2 - 1 = (t - 1) - k'$.
This implies that all the vertices in the path from $s_y$ to $\delta_{k'}$
has received the message.

It remains to argue about the vertices in the path from $\alpha_i$ to $\beta_{i}$.
By the protocol, if $A_i = \{w_i, x_j, y_k\}$,
vertex $\alpha_i$ gets the message at time $(t-\size(x_j))$ from $s_x$
and $\beta_i$ gets the message at time $(t - \size(y_k))$ from $s_y$.
These two vertices forward the message towards each other 
to convey it to all the vertices in the path connecting them.
As $T - \lambda \le \size(w_i) + \size(x_j) + \size(y_k)$, which implies
$T - \lambda - \size(w_i) - 2 \le  (\size(x_j) - 1) + (\size(y_k) - 1)$,
and hence all the vertices in the path receive the message in $t$ rounds.
This concludes the proof of the lemma.
\end{proof}

\begin{lemma}
\label{lemma:bounded-fvs-np-hard-backward}
If $(G, s, t)$ is a \yes-instance of 
\textsc{Telephone Broadcast} then
$(W, X, Y, \size, T, \lambda)$ is a \yes-instance of 
\textsc{Numerical $3$-Dimensional (Almost) Matching}.
\end{lemma}
\begin{proof}
Consider a broadcasting protocol that is represented by
$(\calT,\{C(v) \mid v\in V(\calT)\})$, where $\calT$ is a spanning tree of $G$ 
rooted in $s$ and for each $v \in V(\calT)$, 
$C(v)$ is an ordered set of children of $v$ in $\calT$.
Recall that as soon as $v$ gets the message, $v$ starts to send it to the children 
in $\calT$ in the order defined by $C(v)$. 
We establish specific properties of this broadcasting.

First, consider $C(s)$, i.e., ordering of children of $s$.
Consider vertex $\gamma_{j'}$ for $j' = 1$.
This vertex is at distance $(t-2) - (j'-1) = t - 2$ from $s_x$,
and hence at distance $t$ from $s_x$.
By Observation~\ref{obs:message-propelling},
there is a path from $s$ to $\gamma_1$ with zero-downtime.
As there is a unique shortest path from $s$ to $\gamma_1$,
this path has zero downtime; this implies that in the first round, $s$ forwards the message towards $s_x$. 
The distance of $\delta_{k'}$ from $s_y$ is $(t - 1) - k'$
and hence from $s$ is $(t-1) - k' + 1$.
This implies $\delta_1$ is at distance $t - 1$ from $s$.
By Observation~\ref{obs:message-propelling},
there is a path from $s$ to $\gamma_1$ with downtime of one.
There is unique path from $s$ to $\delta_1$ of length $t - 1$,
which is vertex disjoint from $s$ to $\gamma_1$ path, apart from $s$.
As there is downtime on the path from $s$ to $\delta_1$ in the first round,
in the second round $s$ should forward the message to $s_y$.
Hence, by the end of the second round, both $s_x$ and $s_y$ have
received the message.

Now, consider $C(s_x)$, i.e., the ordering in which $s_x$ forwards
the message to its children. 
The number of children of $s_x$ is $t - 2$ and they are either 
$\gamma$-type vertices or $\alpha$-type vertices.
Consider any two vertices $\gamma$-type vertices
$\gamma_{j'}$ and $\gamma_{j''}$ such that
$j' < j''$, and hence $\dist(s_x, \gamma_{j''}) < \dist(s_x, \gamma_{j'})$.
Suppose $\gamma_{j''}$ appears before $\gamma_{j'}$ 
in the ordering $C(s_x)$, i.e., 
$s_x$ forwards the message
towards $\gamma_{j''}$ at time $t_1$ and towards $\gamma_{j'}$
at time $t_2$ such that $t_1 < t_2$.
As every vertex receives the message in time $t$ and from the
above inequalities,
we have, $\dist(s_x, \gamma_{j''}) < \dist(s_x, \gamma_{j'}) \le t - (t_2 - 1) < t - (t_1 - 1)$.
Hence, exchanging the position of $\gamma_{j''}$ and $\gamma_{j'}$
in $C(s_x)$ also leads to valid protocol.
For the rest of the proof, we suppose that for every $j' < j''$,
$\gamma_{j'}$ appears before $\gamma_{j''}$ in ordering $C(s_x)$.

With such modifications (if needed), suppose 
$s_x$ forwards the message towards $\gamma_{j'}$ at time $f_t(j')$.
Once again, this implies that 
$\dist(s_x, \gamma_{j'}) \le t - (f_t(j') - 1)$.
As $\dist(s_x, \gamma_{j'}) = (t - 2) - (j' - 1)$, we have
$(f(j') - 2) \le j'$.
We use this inequality to determine at what time $s_x$ forwards
the message towards $\gamma_{j'}$.
For example, this implies that when $j' = 1$, the only feasible value for 
$f(j') = 3$ as $s_x$ receives the message by the end of second round.
Hence, $s_x$ should forward the message towards
$\gamma_{1}$ at time $f(j')  = 3$.
Alternately, $\gamma_1$ is the first vertex in $C(s_x)$.
Recall that $\size(x_j)$ is a multiple of $m^2$.
Hence, a similar argument implies that (at least)
first $m^2 - 1$ vertices in $C(s_x)$ are $\gamma_1, \gamma_2, \dots,
\gamma_{m^2-1}$.
Note that if all the indices of $\gamma$-type vertices were consecutive,
then $f(j') = j' - 2$ and the ordering would have continued in the same
fashion.
As this is not the case, in the next paragraph, 
we argue that for a missing index in $\gamma$-type
vertices, $s_x$ forwards the message to $\alpha$-type vertices.  

Suppose $j'$ is the smallest and $j'' (\ge j')$ is the largest 
indices of $\gamma$-type vertices 
such that for every $j^{\circ} \in \{j', j' + 1, \dots, j''\}$,
we have $j^{\circ}$ is in $\{1, 2, \dots, t-2\} \setminus 
\{ t - \size(x_j) \mid j \in [m]\}$, i.e., it is an index of
some $\gamma$-type vertex and
$f_t(j^{\circ}) - 2 < j^{\circ}$.
As $j'$ is the smallest such index, for every $j^{\star} \in \{1, 2, \dots j' - 1\}\setminus \{t - \size(x_j) \mid j \in [m]\}$, 
we have $f_t(j^{\star}) - 2 = j^{\star}$,
more specifically $f_t(j' - 1) - 2 = (j' - 1)$.
As $s_x$ forwards the message towards $\gamma_{j'}$ in the next round,
we have $f_t(j') = f_t(j'-1) + 1$, and hence 
$f_t(j') = j' + 1$.
As next consecutive rounds, $s_x$ sends message towards $\gamma$-type
vertices, we have $f_t(j^{\circ}) - 1 = j^{\circ}$
for every $j^{\circ} \in \{j', j'+1, \dots, j''\}$.
As $j''$ is largest such index, at time $j'' + 1$, the $s_x$ forwards
the message to $\alpha_i$ for some $i \in [m]$.

We consider the following modification to $C(s_x)$.
The $\alpha_i$ is removed from the ordering and inserted
before $j'$. 
The rest of the ordering remains the same.
We argue that this ordering also corresponds to a valid
broadcasting protocol.
With this modification, if in the earlier protocol
$s_x$ forwarded message towards $\gamma_{j^{\circ}}$ 
at time $t_1$, then now it forwards the message 
towards $\gamma_{j^{\circ}}$ at time $t_1 + 1$.
Hence, in this new protocol, $f_t(j^{\circ}) - 2 = j^{\circ}$
This implies vertex $\gamma_{j^{\circ}}$ for every 
$j^{\circ} \in \{j', j'+1, \dots, j''\}$ receives the message in $t$ round.
As the message is send to $\alpha_i$ earlier than before,
any vertex in the path $\alpha_i$ to $\beta_i$ will receive the message
even in this new protocol.
As the remaining order remains unchanged, the modification
results in valid broadcasting protocol.
Repeating this process, we get a new protocol 
in which for every $j' \in \{1, 2, \dots, t-2\} \setminus 
\{t - \size(x_j) \mid j \in [m]\}$,
the vertex $s_x$ forwards the message to one of the 
$\gamma$-type vertex in round $j' + 2$.
As $s_x$ has $(t - 2)$ many children and it receives the message
by the end of the second round,
for every time $j' \in \{t - \size(x_j) \mid j \in [m]\}$,
it sends the message to $\alpha$-type vertex in round $j' + 2$.

Using the similar arguments, for every $k' \in \{1, 2, \dots, t-2\} \setminus 
\{t - \size(x_j) \mid j \in [m]\}$,
the vertex $s_y$ forwards the message to one of the 
$\gamma$-type vertex in round $k' + 2$ and 
for every time $k' \in \{t - \size(x_j) \mid j \in [m]\}$,
it sends the message to $\alpha$-type vertex in round $k'+2$.

Now, suppose vertex $\alpha_i$ gets the message at time $(t-\size(x_j))$
and $\beta_i$ gets the message at time $(t - \size(y_k))$.
These two vertices forward the message towards each other 
to convey it to all the vertices in the path connecting them.
As all the vertices in the path will be informed, we have
$T - \lambda - \size(w_i) - 2 \le  (\size(x_j) - 1) + (\size(y_k) - 1)$ 
which implies 
$T - \lambda \le \size(w_i) + \size(x_j) + \size(y_k)$.
As this is true for every $i \in [m]$, 
this implies 
a solution of \textsc{Telephone Broadcasting} and that of
\textsc{Numerical $3$-Dimensional (Almost) Matching}.
This concludes the proof of the lemma.
\end{proof}

Note that, as \textsc{Numerical $3$-Dimensional (Almost) Matching}
is strongly \NP-complete, the reduction is completed in the time polynomial
in the size of input. 
Lemma~\ref{lemma:bounded-fvs-np-hard-forword} and Lemma~\ref{lemma:bounded-fvs-np-hard-backward} imply
the correctness of the reduction.
This along with the fact that deleting $s_x$ (or $s_y$) removes all the 
cycle from the graph prove Theorem~\ref{thm:fvs-np-hard}.
Also, deleting both $s_x$ and $s_y$ results in the collection of paths.
For Thorem~\ref{thm:tree-depth-eth}, note that the 
treedepth of the path of length $q$ is $\calO(\log(q))$.
This, along with the arguments presented in the proof of 
Theorem~\ref{thm:sol-size-lb}, prove Theorem~\ref{thm:tree-depth-eth}.

\section{Conclusion}
\label{sec:conclusion}
In this article, we studied the \textsc{Telephone Broadcast} problem 
and answered two open question by Fomin et al.~\cite{DBLP:conf/wg/FominFG23}.
Our reductions results in somewhat rare results.
First, the problem admits a double exponential lower bound 
when parameterized by the solution size under the \ETH.
Second, we prove that the problem is \NP-\complete\ even
on graphs of feedback vertex set number one, and hence
on graphs of tree-width at most two.
Hence, this result proves very tight polynomial vs \NP-\complete\ 
dichotomy separating treewidth one from treewidth two, which is a rare phenomenon.
The same reduction also implies that the problem does not
admit an algorithm running in time $2^{2^{o(\td)}}\cdot n^{\calO(1)}$
unless the \ETH\ fails.

Fomin et al.~\cite{DBLP:conf/wg/FominFG23} presented the algorithm
that runs in time $2^{2^{\calO(\vc)}}\cdot n^{\calO(1)}$.
Note that tree-depth is smaller parameter than vertex cover.
Hence, we have a lower bound with respect to the tree-width 
but upper bound with respect to the vertex cover number.
It would be interesting to close this gap either by improving the 
algorithm or tightening the lower bound.

\section*{Addendum}
After sharing the manuscript, Prof. Dr. Petr Golovach pointed
out that Papadimitriou and Yannakakis~\cite{DBLP:journals/jacm/PapadimitriouY82} mentioned that 
the problem of determining whether a graph contains a binomial
spanning tree is \NP-\complete~(Theorem~8). 
Consider a reduction that given an instance $(G)$,
a graph on $n$ vertices, of this problem,
constructs an instance $(G', s, t)$ of \textsc{Telephone Broadcast}
by adding a global vertex $s$ to $G$ and making it adjacent with
$t - 1$ pendant vertices, where $t = \calO(\log(n))$.
Hence, the \textsc{Telephone Broadcast} problem is \NP-\complete\
even when $|V(G')| = \calO(2^t)$.
As the original problem is \NP-\complete,
this implies a stronger statement than Theorem~\ref{thm:sol-size-lb}. 
 



\bibliography{references.bib}

\end{document}